\documentclass[conference]{IEEEtran}

\PassOptionsToPackage{table}{xcolor}
\usepackage{cite}
\usepackage{url}
\usepackage[table]{xcolor}
\usepackage{colortbl}
\usepackage{paralist}
\usepackage[utf8]{inputenc} 
\usepackage{fixltx2e} 
\usepackage[T1]{fontenc} 
\usepackage{xurl}
\usepackage{hyperref}
\usepackage{graphicx}
\usepackage{tikz,tabularx} 
\usepackage{amsmath}

\usepackage{amssymb}
\usepackage{amsfonts}
\usetikzlibrary{decorations.pathreplacing,chains,fit,shapes,calc,arrows,plotmarks,arrows.meta,shadows,positioning,shapes.arrows,shapes.geometric,shapes.multipart,calligraphy,patterns, patterns.meta}
\usepackage{textcomp}
\usepackage{pgfplots}
\usepackage{svg}
\usepackage[binary-units=true,per-mode=symbol]{siunitx}
\usepackage[nolist]{acronym}
\usepackage{booktabs}
\usepackage{algorithm}
\usepackage{algpseudocode}
\usepackage{textcomp}
\usepackage{listings}
\usepackage{csquotes}
\usepackage{stmaryrd}
\usepackage{caption}
\usepackage{subcaption}
\usepackage{fancyvrb}
\usepackage{fontawesome}
\usepackage{multirow,multicol,csvsimple}
\usepackage[inline]{enumitem}
\usepackage{amsthm}
\newtheorem{definition}{Definition}
\newtheorem{lemma}{Lemma}
\usepackage{pgfplotstable}
\usepackage{dblfloatfix}
\usepackage{afterpage}

\usepackage{tocbibind}
\usepackage{placeins}

\usepackage{pdfpages}
\usepackage[disable]{todonotes}
\usepackage{longtable}
\usepackage{pdflscape}
\usepackage{adjustbox}
\usepackage{xspace}
\usepackage{makecell,rotating}
\usepackage[nameinlink,noabbrev,capitalize]{cleveref}
\usepackage{comment}
\usepackage{soul}
\usepackage{orcidlink}

\newboolean{showcomments}
\setboolean{showcomments}{false}

\ifthenelse{\boolean{showcomments}}
{ \newcommand{\mynote}[3]{
    \protect\fbox{\bfseries\sffamily\scriptsize#1}
    {\small\textsf{\emph{\color{#3}{#2}}}}}}
{ \newcommand{\mynote}[3]{}}

\newcommand\important[1]{\todo[inline]{\textbf{Important:} #1}}
\newcommand\leo[1]{\todo[color=orange,inline]{\textbf{Leo:} #1}}
\newcommand\marcel[1]{\todo[color=green,inline]{\textbf{Marcel:} #1}}

\ifthenelse{\boolean{showcomments}}{}{
    \renewcommand\important[1]{}
    \renewcommand\leo[1]{}
    \renewcommand\marcel[1]{}
}

\newcolumntype{R}[2]{%
    >{\adjustbox{angle=#1,lap=\width-(#2)}\bgroup}%
    l%
    <{\egroup}%
}

\newcommand{\cf}{cf.\@\xspace}

\newcommand{\etal}{\textit{et al.}\@\xspace}
\newcommand{\eg}{\textit{e.g.}\@\xspace}

\newcommand{\ie}{\textit{i.e.}\@\xspace}

\begin{acronym}[derp]
        \acro{AE}[AE]{Asymmetric Extremum}
        \acro{RAM}[RAM]{Rapid Asymmetric Extremum}
        \acro{LMC}[LMC]{Local Maximum Chunking}
        \acro{FBC}[FBC]{Frequency-Based Chunking}
        \acro{BFBC}[BFBC]{Bytes-Frequency-Based Chunking}
        \acro{MII}[MII]{Minimal Incremental Interval}
        \acro{PCI}[PCI]{Parity Check of Interval}
        \acro{MUCH}[MUCH]{Multithreaded Content-Based File Chunking}
        \acro{CDC}[CDC]{Content-Defined Chunking}
        \acro{FSC}[FSC]{Fixed-Size Chunking}
        \acro{BSW}[BSW]{Basic Sliding Window}
        \acro{ipc}[IPC]{instructions per cycle}
        \acro{iqr}[IQR]{interquartile range}
        \acro{simd}[SIMD]{single instruction multiple data}
        \acro{rtt}[RTT]{round-trip time}
        \acro{NC}[NC]{Normalized Chunking}
        \acro{LMC}[LMC]{Local Maximum Chunking}
        \acro{TTTD}[TTTD]{Two-thresholds–two-divisors}
        \acro{RC}[RC]{Regression Chunking}
\end{acronym}

\begin{document}


\title{
Benchmarking Content-Defined Chunking Algorithms for Data Deduplication
}

\author{
Marcel Gregoriadis$^{1}$\orcidlink{0000-0001-5094-0111},
Leonhard Balduf$^{2}$\orcidlink{0000-0002-3519-7160},
Björn Scheuermann$^{2}$\orcidlink{0000-0002-1133-1775},
Johan Pouwelse$^{1}$\orcidlink{0000-0002-9882-1506}\\[0.5em]
{\small $^{1}$Delft University of Technology, Delft, The Netherlands}\\
{\small $^{2}$Technical University of Darmstadt, Darmstadt, Germany}\\
{\small
m.gregoriadis@tudelft.nl,
leonhard.balduf@tu-darmstadt.de,
scheuermann@tu-darmstadt.de,
j.a.pouwelse@tudelft.nl
}
}


\maketitle

\begin{abstract}
    Data deduplication relies on content-defined chunking (CDC) algorithms to decide where data is split.
We observe that most CDC algorithms are modular: a \emph{base chunking algorithm} fixes the byte-level cut-point condition, while \emph{high-level techniques} (size normalization, min/max thresholds, caching) layer on top. Existing benchmarks conflate the two, obscuring which design choices actually drive performance.
We isolate and evaluate eight state-of-the-art base chunking algorithms across four realistic workloads and four metrics: throughput, deduplication ratio, mean chunk size, and chunk-size variance.
Further, we derive new formulas relating algorithm parameters to expected chunk size for AE, RAM, and MII. 
In doing so, we uncover a longstanding error in the widely-used AE parameterization reproduced across multiple papers and evaluation tools that produces chunks 13--39\,\% smaller than intended and inflates AE's reported deduplication.\footnote{The popular benchmark tool DedupBench has since corrected it following our report.}
Across our evaluation, AE is the most balanced base chunking algorithm across all four metrics, RAM's similarity to AE does not carry over to low-entropy data, and MII and PCI show limited advantages over simpler alternatives in storage-oriented workloads. We release all implementations, datasets, and analysis scripts as open source.

%
%
\end{abstract}


\section{Introduction}\label{sec:intro}
Generative AI is driving rapid data growth and increasing the pressure on storage capacity and network bandwidth.
IDC projects global data generation to double between 2025 and 2029, reaching 527 zettabytes annually~\cite{idc2025}.
Yet across the global datasphere most stored data is redundant, with IDC estimating the ratio of unique to replicated data at roughly 1:9~\cite{idc2018}.
This redundancy is highest in backup and archival workloads, including near-identical VM images, container layers, snapshots, and database backups~\cite{peng2025bbs}, where successive versions overlap heavily. 
Chunk-based data deduplication~\cite{lbfs,quinlan2002venti,xia2016comprehensive} makes retaining and synchronizing this history affordable.
Large-scale studies of Microsoft~\cite{meyer2012study,rabinPrimary} and EMC~\cite{wallace2012characteristics} report space savings of up to \qty{83}{\percent}.
Data deduplication has been studied directly in database and primary-storage settings, including online database deduplication and inline deduplication for primary storage~\cite{xu2017online,wildani2013hands}.
At the core of chunk-based storage deduplication sits a chunking algorithm that partitions an incoming byte stream into chunks for redundancy detection.
While partitioning data into fixed-size chunks is simple and efficient, modern deduplication systems employ \emph{content-defined chunking (CDC)} algorithms that place boundaries based on a chunk's \emph{content} rather than its size.
CDC improves resilience to byte insertions and deletions, thereby increasing the amount of redundant content that can be identified~\cite{lbfs,rabinPrimary}.
The choice of the chunking algorithm directly governs both how much redundancy is exposed and how fast data can be ingested.
Over the years, numerous chunking algorithms have been proposed,
claiming higher throughput~\cite{sscdc,udayashankar2024seqcdc,fastcdc,ram,quickcdc},
improved deduplication ratios~\cite{bimodal,fbc,anchorDrivenSubchunk},
and lower chunk-size variance~\cite{fastcdc2,ae2,tttd}.

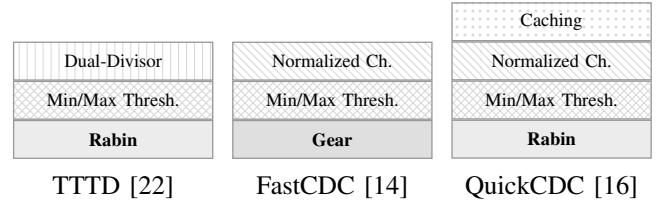
\begin{figure}[!t]
\begin{tikzpicture}[
    box/.style={
        draw,
        shape=rectangle, 
        font=\scriptsize,
        rounded corners=0pt,
        text width=2.4cm,
        minimum height=0.5cm,
        line width=0.5pt,
        align=center,
        fill=gray!15,
        draw=gray!80,
        pattern color=gray!30
    },
    dots/.style={box, pattern=dots},
    lines/.style={box, pattern=north east lines},
    grid/.style={box, pattern=grid},
    crosshatch/.style={box, pattern=crosshatch},
    bricks/.style={box, pattern=bricks},
    checkerboard/.style={box, pattern=checkerboard},
    fivepointed/.style={box, pattern=fivepointed stars},
    sixpointed/.style={box, pattern=sixpointed stars},
    horizontal/.style={box, pattern=horizontal lines},
    vertical/.style={box, pattern=vertical lines},
    north/.style={box, pattern=north west lines},
    solid/.style={box, fill=gray!20},
    node distance=0cm
]
\node[] (tttd) at (0,0) {TTTD~\cite{tttd}};
\node[box, above=0.1cm of tttd] (rabin1) {\textbf{Rabin}};
\node[box, crosshatch, above=of rabin1] (minmax1) {Min/Max Thresh.};
\node[box, vertical, above=of minmax1] (dual) {Dual-Divisor};
\node[] (fastcdc) at (2.9,0) {FastCDC~\cite{fastcdc}};
\node[box, fill=gray!25, above=0.1cm of fastcdc] (gear) {\textbf{Gear}};
\node[box, crosshatch, above=of gear] (minmax2) {Min/Max Thresh.};
\node[box, north, above=of minmax2] (nc1) {Normalized Ch.};
\node[] (quickcdc) at (5.8,0) {QuickCDC~\cite{quickcdc}};
\node[box, above=0.1cm of quickcdc] (rabin2) {\textbf{Rabin}};
\node[box, crosshatch, above=of rabin2] (minmax3) {Min/Max Thresh.};
\node[box, north, above=of minmax3] (nc2) {Normalized Ch.};
\node[box, dots, above=of nc2] (caching) {Caching};
\end{tikzpicture}
\caption{Most modern chunking algorithms can be decomposed into modular building blocks.
}
\label{fig:composite}
\end{figure}

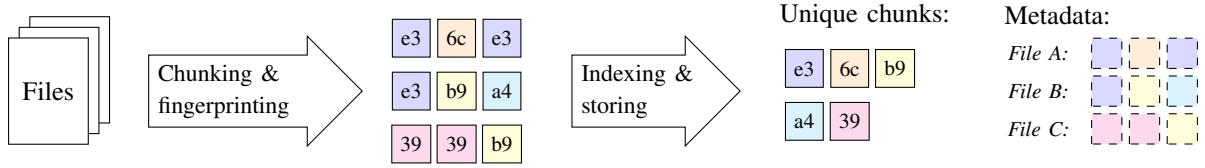
\begin{figure*}[!tb]
    \centering
    \begin{tikzpicture}[
    doc/.style={draw, minimum height=4em, minimum width=3em, 
                fill=white, 
                double copy shadow={shadow xshift=4pt, 
                             shadow yshift=4pt, fill=white, draw}},
]
\tikzset{
myarrow/.style={
  draw,
  single arrow,
  text width=2.1cm,
  fill=white
  },
}

\tikzset{
chunk/.style={
  draw=black,
  fill=blue,
  fill opacity=0.15,
  text opacity=1,
  minimum width=5mm,
  minimum height=5mm,
  align=center,
  font=\footnotesize
  },
}
\tikzset{
metachunk/.style={
  dashed,
  draw=black,
  fill=blue,
  fill opacity=0.15,
  minimum width=4mm,
  minimum height=4mm
  },
}

   \node[doc] (doc) {Files};
   
   \node[myarrow] at (2.5,0) {\small Chunking \& fingerprinting};
   
   \draw (4.8,0.7) node [chunk,fill=blue] {e3} ;
   \draw (5.4,0.7) node [chunk,fill=orange] {6c} ;
   \draw (6,0.7) node [chunk,fill=blue] {e3} ;
   
   \draw (4.8,0) node [chunk,fill=blue] {e3} ;
   \draw (5.4,0) node [chunk,fill=yellow] {b9} ;
   \draw (6,0) node [chunk,fill=cyan] {a4} ;
   
   \draw (4.8,-0.7) node [chunk,fill=magenta] {39} ;
   \draw (5.4,-0.7) node [chunk,fill=magenta] {39} ;
   \draw (6,-0.7) node [chunk,fill=yellow] {b9} ;
   
   \node[myarrow,text width=1.5cm] at (7.8,0) {\small Indexing \& storing};
   
   \node[] at (10.8,1) {Unique chunks:};
   
   \draw (10.0,0.3) node [chunk,fill=blue] {e3} ;
   \draw (10.6,0.3) node [chunk,fill=orange] {6c} ;
   \draw (11.2,0.3) node [chunk,fill=yellow] {b9} ;
   
   \draw (10.0,-0.4) node [chunk,fill=cyan] {a4} ;
   \draw (10.6,-0.4) node [chunk,fill=magenta] {39} ;
   
   \node[] at (13.35,1) {Metadata:};
   
   \node[] at (13.1,0.5) {\footnotesize \textit{File A:}};
   \draw (13.8,0.3) [metachunk,fill=blue] rectangle ++(0.4,0.4);
   \draw (14.3,0.3) [metachunk,fill=orange] rectangle ++(0.4,0.4);
   \draw (14.8,0.3) [metachunk,fill=blue] rectangle ++(0.4,0.4);
   
   \node[] at (13.1,0) {\footnotesize \textit{File B:}};
   \draw (13.8,-0.2) [metachunk,fill=blue] rectangle ++(0.4,0.4);
   \draw (14.3,-0.2) [metachunk,fill=yellow] rectangle ++(0.4,0.4);
   \draw (14.8,-0.2) [metachunk,fill=cyan] rectangle ++(0.4,0.4);
   
   \node[] at (13.1,-0.5) {\footnotesize \textit{File C:}};
   \draw (13.8,-0.7) [metachunk,fill=magenta] rectangle ++(0.4,0.4);
   \draw (14.3,-0.7) [metachunk,fill=magenta] rectangle ++(0.4,0.4);
   \draw (14.8,-0.7) [metachunk,fill=yellow] rectangle ++(0.4,0.4);

   
   
\end{tikzpicture}
    \caption{Overview of data deduplication process.}
    \label{fig:dedupProcess}
\end{figure*}

We find that many proposed chunking algorithms are compositions of interchangeable techniques. Based on this observation, we advocate a \emph{modular} view of chunking algorithms, distinguishing between \emph{base chunking algorithms} (\eg, Rabin~\cite{lbfs} and Gear~\cite{ddelta}) and \emph{composite algorithms} (\eg, TTTD~\cite{tttd} and FastCDC~\cite{fastcdc2}), which augment base chunking algorithms with higher-level optimizations.
\Cref{fig:composite} illustrates three pervasive chunking algorithms and their decompositions.
This perspective reveals that many composite algorithms are evaluated only in combination with a specific base chunking algorithm, despite being compatible with a broader range of alternatives. Consequently, identifying the strongest base chunking algorithm becomes a fundamental question. Yet the state of the art remains unclear. 
Existing studies often conflate the effects of high-level techniques and base chunking algorithms rather than evaluating them independently~\cite{udayashankar2024seqcdc,udayashankar2025vectorcdc,ram,udayashankar2024impact}. 
For example, TTTD is frequently benchmarked using Rabin~\cite{liu2023dedupbench,udayashankar2024seqcdc}, while newer composite algorithms such as RapidCDC~\cite{rapidcdc}, UltraCDC~\cite{zhou2022ultracdc}, and SuperCDC~\cite{wan2022supercdc} largely recombine existing techniques in different ways. As a result, current benchmarks provide limited insight into the relative merits of the underlying base chunking algorithms.
The problem is further exacerbated by inconsistent evaluation methodologies. Existing benchmarks are typically conducted on different datasets and with custom metrics~\cite{ddelta,ae2,ram,mii,pci}. Moreover, several algorithms lack a principled method for selecting parameters that achieve a desired target chunk size, precluding fair comparisons based on deduplication efficiency, since smaller chunks generally yield higher deduplication ratios.
For instance, we uncover a longstanding parametrization error in AE~\cite{ae2} that caused chunks to be 13--39\,\% smaller than the configured target size, artificially inflating reported deduplication ratios.

In this paper, we conduct a rigorous analysis and benchmark of representative base chunking algorithms.
Our goal is to determine which byte-level cut-point rules remain effective once high-level techniques are stripped away and algorithm settings are selected for common configured targets on uniformly random data before being evaluated unchanged on realistic workloads.
Specifically, this paper makes the following contributions:
\begin{itemize}
    \item We release a new benchmark framework that supports modular CDC evaluations by separating base chunking algorithms from high-level techniques.
    \item We benchmark eight representative base chunking algorithms---Rabin~\cite{rabinPrimary}, Buzhash~\cite{buzhash}, Gear~\cite{ddelta}, AE~\cite{ae2}, RAM~\cite{ram}, MII~\cite{mii}, PCI~\cite{pci}, and SeqCDC~\cite{udayashankar2024seqcdc}---on four realistic workloads (\Cref{table:datasets}), measuring throughput, average chunk size, chunk-size variance, and deduplication ratio.
    \item We derive parameter rules for common target chunk sizes, including formulas for AE, RAM, and MII, and assess their calibration on uniformly random data, distinguishing configured targets from empirical means.
\end{itemize}


\section{Background and Motivation}\label{sec:background}
Data deduplication splits files into chunks, fingerprints them, creates reconstruction metadata, and stores unique chunks and metadata (see \Cref{fig:dedupProcess}).
Our work focuses on the first stage of this pipeline.
The algorithm by which files are chunked directly influences the probability of finding redundant chunks, which can then be deduplicated.
The most efficient approach, \ac{FSC}, splits a file into equal-sized chunks but suffers from the \emph{boundary-shift problem} illustrated in \Cref{fig:boundaryShiftProblem}: inserted or removed content shifts boundaries and causes unchanged regions to appear unique.
\ac{CDC} algorithms address this problem by forming chunk boundaries based on \emph{content} rather than \emph{position}.
These failure modes and design goals motivate the main CDC families discussed below.

\begin{figure}[h]
    \centering
    \begin{tikzpicture}
    \node[label] (cbl) at (0.0,1.7) {\emph{chunk boundary}};
    \node (cb) at (1.5, 1.36) {};
    \node[label] (shifted) at (2.5, 0.6) {\emph{shifted}};
    \node[label] (inserted) at (2.1, -0.55) {\emph{inserted content}};
    \draw[->] (cbl.east) -- (cb.north);
    \draw[->] (1.5, 0.8) -- (2, 0.3);
    \draw[->] (inserted.west) -- (0.5, -0.1);
    \node at (-0.8,1.1) {File A};
    \draw[thick] (1.5, 1.35) -- (1.5, 0.85);
    \draw[thick] (3, 1.35) -- (3, 0.85);
    \draw[thick] (4.5, 1.35) -- (4.5, 0.85);
    \draw[draw=black] (0, 0.95) rectangle ++(6, 0.3);
    \node at (-0.8,0.05) {File B};
    \draw[thick] (1.5, 0.3) -- (1.5, -0.15);
    \draw[thick] (3, 0.3) -- (3, -0.15);
    \draw[thick] (4.5, 0.3) -- (4.5, -0.15);
    \draw[thick,dotted] (2, 0.3) -- (2, -0.15);
    \draw[thick,dotted] (3.5, 0.3) -- (3.5, -0.15);
    \draw[thick,dotted] (5, 0.3) -- (5, -0.15);
    \draw[draw=black] (0, -0.1) rectangle ++(6, 0.3);
    \draw[draw=black,fill=lightgray] (0, -0.1) rectangle ++(0.5, 0.3);
\end{tikzpicture}
    \vspace{-1.5em}
    \caption{Boundary-shift problem.}
    \label{fig:boundaryShiftProblem}
\end{figure}
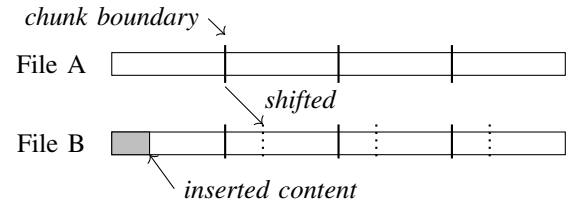


\subsection{Hash-Based CDC}
\label{subsec:background:hash-based-cdc}

The earliest \ac{CDC} algorithms used a sliding window of size $w$ over a file
and a hash function
to determine chunk boundaries.
The hash of each iteration of the sliding window is compared against a bitmask to determine new chunk boundaries.
For efficiency this is done using \emph{rolling hash} functions, such as Rabin~\cite{rabin,rabinkarp,rabinApplications}, Buzhash~\cite{buzhash}, or Gear~\cite{ddelta}.
We refer to this approach generally as \ac{BSW}, illustrated in \Cref{fig:slidingWindow}.
Hash-based CDC yields higher deduplication ratios than \ac{FSC}, but is criticized for low throughput and high chunk-size variance.
Efficient rolling hashes such as Gear~\cite{ddelta} mitigate the former, while the latter remains intrinsic to hash-based CDC.

\begin{figure}[tb]
  \centering
  \begin{tikzpicture}[scale=0.9, every node/.style={scale=0.9}]
\tikzstyle{every path}=[thick]

\edef\sizetape{0.6cm}
\tikzstyle{tmtape}=[draw,minimum size=\sizetape]
\tikzstyle{tmhead}=[arrow box,draw,minimum size=.5cm,arrow box
arrows={east:.25cm, west:0.25cm}]


\begin{scope}[start chain=1 going right,node distance=-0.15mm,local bounding box=tape1]
    \draw[draw=black,line width=1.8pt] (-0.29,-0.29) rectangle ++(2.45,0.6);
    \node [on chain=1,tmtape] (start) {h};
    \node [on chain=1,tmtape] {e};
    \node [on chain=1,tmtape] {l};
    \node [on chain=1,tmtape] (end) {l};
    \node [on chain=1,tmtape] {o};
    \node [on chain=1,tmtape] {};
    \node [on chain=1,tmtape] {w};
    \node [on chain=1,tmtape] {o};
    \node [on chain=1,tmtape] {r};
    \node [on chain=1,tmtape] {l};
    \node [on chain=1,tmtape] {d};
    \draw [decorate,decoration={brace,amplitude=5pt,mirror,raise=3ex}]
        (start.west) -- (end.east) node[midway,yshift=-2.8em]{sliding window $w$};
\end{scope}
  

\begin{scope}[start chain=1 going right,node distance=-0.15mm,local bounding box=tape2, shift={($(tape1.west)-(-0.5,1.2cm)$)}]
    \draw[draw=black,line width=1.8pt] (0.32,-0.31) rectangle ++(2.45,0.6);
    \node [on chain=1,tmtape] {h};
    \node [on chain=1,tmtape] {e};
    \node [on chain=1,tmtape] {l};
    \node [on chain=1,tmtape] {l};
    \node [on chain=1,tmtape] {o};
    \node [on chain=1,tmtape] {};
    \node [on chain=1,tmtape] {w};
    \node [on chain=1,tmtape] {o};
    \node [on chain=1,tmtape] {r};
    \node [on chain=1,tmtape] {l};
    \node [on chain=1,tmtape] {d};
\end{scope}
  

\begin{scope}[start chain=1 going right,node distance=-0.15mm, shift={($(tape2.west)-(-0.31,1cm)$)}]
    \node [draw=none,yshift=14,xshift=87] {$\ldots$};
    \draw[draw=black,line width=1.8pt] (4.01,-0.31) rectangle ++(2.45,0.6);
    \node [on chain=1,tmtape] {h};
    \node [on chain=1,tmtape] {e};
    \node [on chain=1,tmtape] {l};
    \node [on chain=1,tmtape] {l};
    \node [on chain=1,tmtape] {o};
    \node [on chain=1,tmtape] {};
    \node [on chain=1,tmtape] {w};
    \node [on chain=1,tmtape] {o};
    \node [on chain=1,tmtape] {r};
    \node [on chain=1,tmtape] {l};
    \node [on chain=1,tmtape] {d};
\end{scope}

\end{tikzpicture}
  \vspace{1em}
  \caption{Window sliding over a file byte-by-byte. At every offset, the fingerprint generated from the current window is matched against a bitmask.}
  \label{fig:slidingWindow}
\end{figure}
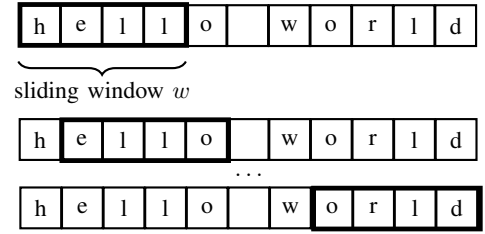

\subsection{Chunk-Size Variance}
\label{subsec:background:chunk-size-variance}

High chunk-size variance gives rise to the issue of pathological, \ie very large or very small, chunk sizes~\cite{tttd}.
Very large chunks can be the product of a recurring pattern in the byte sequence
that fails to meet the criteria for setting a cut-point~\cite{schleimer2003winnowing}.
Large chunks are less likely to be duplicates and more likely to be affected by modifications, increasing the number of bytes that must be stored again~\cite{tttd}.
Pathologically small chunks are undesirable as well
because they produce more overhead due to additional metadata, processing,
and, in distributed settings, additional round-trip times.
In some use cases, high chunk-size variance is not a critical concern.
For example, with MII~\cite{mii} and PCI~\cite{pci}, Zhang \etal focus on the application of CDC for \emph{incremental file synchronization} (\eg, with rsync~\cite{rsync}).
This use case does not consider data reduction within long-term storage systems but instead a transaction of file synchronization between two machines.
This relaxes the constraint on low chunk-size variance,
as the overhead from excessively large or small chunks does not accumulate over multiple revisions of the dataset.

\subsection{Hashless CDC}
\label{subsec:background:hash-less-cdc}

High chunk-size variance and low throughput led to the development of other approaches.
MAXP~\cite{maxp} and LMC~\cite{lmc} were the first CDC algorithms to determine chunk boundaries based on local extrema in the data stream.
As an illustrative example, \Cref{fig:ae} presents AE~\cite{ae}, a representative hashless CDC algorithm.
By using byte comparisons rather than hash functions, these algorithms claim to achieve higher throughput than hash-based algorithms.
Furthermore, they are attributed with a significantly lower chunk-size variance~\cite{lmc,ae}.
Modern adaptations of this technique can be found in AE~\cite{ae} and RAM~\cite{ram}.
Alternative hashless CDC algorithms exist:
PCI~\cite{pci} compares the popcount of a sliding window to a threshold in order to determine chunk boundaries.

\begin{figure}[ht]
  \centering
  \scalebox{0.9}{
\begin{tikzpicture}

\tikzstyle{line} = [draw, latex'-latex']

\definecolor{darkgray176}{RGB}{176,176,176}
\definecolor{steelblue31119180}{RGB}{31,119,180}

\draw (0.83,2.65) node[draw=none,color=red!80!black,align=center] {$c_1$};
\path [->,draw=red!80!black,dotted,thick] (0.59,1.285) -- node [text width=1.2cm,color=red!80!black,midway,above,align=center] {\footnotesize $h$} (1.56,1.285);
\path [draw=red!80!black,thick] (1.57,0) -- node [] {} (1.57,2.42);

\draw (2.805,2.65) node[draw=none,color=red!80!black,align=center] {$c_2$};
\path [->,draw=red!80!black,dotted,thick] (2.8,1.905) -- node [text width=1.2cm,color=red!80!black,midway,above,align=center] {} (3.74,1.905);
\path [draw=red!80!black,thick] (3.74,0) -- node [] {} (3.74,2.42);

\draw (4.845,2.65) node[draw=none,color=red!80!black,align=center] {$c_3$};
\path [->,draw=red!80!black,dotted,thick] (4.48,2.16) -- node [text width=1.2cm,color=red!80!black,midway,above,align=center] {} (5.42,2.16);
\path [draw=red!80!black,thick] (5.44,0) -- node [] {} (5.44,2.42);

\draw (6.38,2.65) node[draw=none,color=red!80!black,align=center] {$c_4$};
        
\begin{axis}[
tick align=outside,
tick pos=left,
xlabel = {Byte position},
ylabel = {Byte value},
x grid style={darkgray176},
xmin=-1, xmax=28,
xtick style={color=black},
y grid style={darkgray176},
ymin=0, ymax=256,
ytick style={color=black},
height=4cm,width=8.6cm
]
\draw[draw=none,fill=steelblue31119180] (axis cs:-0.4,0) rectangle (axis cs:0.4,30);
\draw[draw=none,fill=steelblue31119180] (axis cs:0.6,0) rectangle (axis cs:1.4,138);
\draw[draw=none,fill=steelblue31119180] (axis cs:1.6,0) rectangle (axis cs:2.4,61);
\draw[draw=none,fill=steelblue31119180] (axis cs:2.6,0) rectangle (axis cs:3.4,117);
\draw[draw=none,fill=steelblue31119180] (axis cs:3.6,0) rectangle (axis cs:4.4,16);
\draw[draw=none,fill=steelblue31119180] (axis cs:4.6,0) rectangle (axis cs:5.4,44);
\draw[draw=none,fill=steelblue31119180] (axis cs:5.6,0) rectangle (axis cs:6.4,110);
\draw[draw=none,fill=steelblue31119180] (axis cs:6.6,0) rectangle (axis cs:7.4,194);
\draw[draw=none,fill=steelblue31119180] (axis cs:7.6,0) rectangle (axis cs:8.4,186);
\draw[draw=none,fill=steelblue31119180] (axis cs:8.6,0) rectangle (axis cs:9.4,123);
\draw[draw=none,fill=steelblue31119180] (axis cs:9.6,0) rectangle (axis cs:10.4,203);
\draw[draw=none,fill=steelblue31119180] (axis cs:10.6,0) rectangle (axis cs:11.4,104);
\draw[draw=none,fill=steelblue31119180] (axis cs:11.6,0) rectangle (axis cs:12.4,4);
\draw[draw=none,fill=steelblue31119180] (axis cs:12.6,0) rectangle (axis cs:13.4,89);
\draw[draw=none,fill=steelblue31119180] (axis cs:13.6,0) rectangle (axis cs:14.4,27);
\draw[draw=none,fill=steelblue31119180] (axis cs:14.6,0) rectangle (axis cs:15.4,98);
\draw[draw=none,fill=steelblue31119180] (axis cs:15.6,0) rectangle (axis cs:16.4,98);
\draw[draw=none,fill=steelblue31119180] (axis cs:16.6,0) rectangle (axis cs:17.4,230);
\draw[draw=none,fill=steelblue31119180] (axis cs:17.6,0) rectangle (axis cs:18.4,21);
\draw[draw=none,fill=steelblue31119180] (axis cs:18.6,0) rectangle (axis cs:19.4,192);
\draw[draw=none,fill=steelblue31119180] (axis cs:19.6,0) rectangle (axis cs:20.4,90);
\draw[draw=none,fill=steelblue31119180] (axis cs:20.6,0) rectangle (axis cs:21.4,128);
\draw[draw=none,fill=steelblue31119180] (axis cs:21.6,0) rectangle (axis cs:22.4,100);
\draw[draw=none,fill=steelblue31119180] (axis cs:22.6,0) rectangle (axis cs:23.4,27);
\draw[draw=none,fill=steelblue31119180] (axis cs:23.6,0) rectangle (axis cs:24.4,154);
\draw[draw=none,fill=steelblue31119180] (axis cs:24.6,0) rectangle (axis cs:25.4,128);
\draw[draw=none,fill=steelblue31119180] (axis cs:25.6,0) rectangle (axis cs:26.4,20);
\draw[draw=none,fill=steelblue31119180] (axis cs:26.6,0) rectangle (axis cs:27.4,149);
\end{axis}

\end{tikzpicture}}
  \caption{Illustration of the AE chunking algorithm on a sequence of 27 bytes and a horizon $h=4$. The vertical red lines mark the cut points.
  }
  \label{fig:ae}
\end{figure}
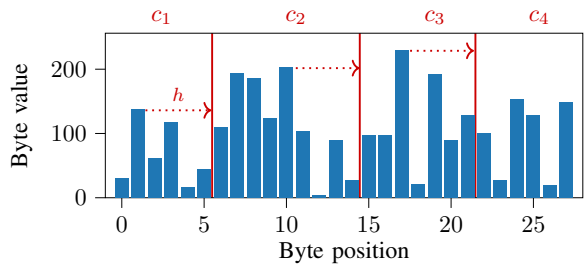

\begin{table*}[ht]
    \centering
    \caption{Compatibility matrix between base chunking algorithms and high-level techniques.}
    \renewcommand{\arraystretch}{1.1}
    \setlength{\tabcolsep}{6pt}
    \begin{tabularx}{\textwidth}{Xccccccccc}
        \toprule
        & \multicolumn{9}{c}{\textbf{High-Level Techniques}} \\
        \cmidrule(lr){2-10}
        & \multicolumn{6}{c}{\textbf{Throughput}} & \multicolumn{2}{c}{\textbf{Chunk-Size Variance}} & \textbf{Deduplication} \\
        \cmidrule(lr){2-7} \cmidrule(lr){8-9} \cmidrule(lr){10-10}
        \textbf{\makecell{Base Algorithm}}
         & \makecell{JC\\\cite{jin2023accelerating}}
         & \makecell{SIMD\\\cite{sscdc}}
         & \makecell{MUCH\\\cite{much}}
         & \makecell{VectorCDC\\\cite{udayashankar2025vectorcdc}}
         & \makecell{GPU\\\cite{shredder}}
         & \makecell{Caching\\\cite{quickcdc,rapidcdc}}
         & \makecell{Multiple Thresholds\\\cite{tttd,rabinPrimary,fastcdc2,zhou2022ultracdc,wan2022supercdc}}
         & \makecell{Merging\\\cite{much,tttd,sscdc}}
         & \makecell{Rechunking\\\cite{fbc,bimodal,anchorDrivenSubchunk,zhou2013hysteresis,fingerdiff}} \\
        \midrule
        \textbf{Hash-based} & \checkmark & \checkmark & \checkmark & & \checkmark & \checkmark & \checkmark & \checkmark & \checkmark \\
        \textbf{Hashless} & & & \checkmark & \checkmark & & \checkmark & \checkmark &  \checkmark & \checkmark \\
        \bottomrule
    \end{tabularx}
    \label{tab:cdcmatrix}
\end{table*}

\subsection{High-Level Techniques}
\label{subsec:background:high-level-techniques}

Parallel to improvements in the basic chunking mechanism,
researchers have proposed high-level methods to increase chunking \emph{throughput}, lower \emph{chunk-size variance}, or enhance the \emph{deduplication ratio}.
Often (\eg, with FastCDC~\cite{fastcdc2} and QuickCDC~\cite{quickcdc}), proposed algorithms are a combination of several high-level techniques.
\Cref{tab:cdcmatrix} summarizes their compatibility with base algorithms beyond those used in the original papers.

\paragraph*{Chunking Throughput}
A popular strategy for increasing chunking throughput and lowering chunk-size variance at the same time
is the enforcement of minimum and maximum chunk sizes~\cite{ram,fastcdc2,udayashankar2024seqcdc}.
This ensures that chunk sizes remain within a desirable range
and allows implementations to skip computation for the minimum chunk length within each new chunk.
While effective at improving those two metrics, it contradicts the very idea of CDC:
to set chunk boundaries based on content, not position.
Inevitably, this technique trades off deduplication efficiency~\cite{rabinPrimary,fastcdc2,ae}.
More sophisticated techniques include, \eg,
a caching layer~\cite{quickcdc,rapidcdc}
or normalization of the chunk-size distribution~\cite{fastcdc2,zhou2022ultracdc,wan2022supercdc}
to allow for skipping with no or negligible impact on deduplication ratio.
Other efforts have explored parallelization by exploiting properties of modern hardware~\cite{sscdc,udayashankar2025vectorcdc,shredder,much}.

\paragraph*{Chunk-Size Variance}
Hash-based chunking algorithms suffer from a highly skewed chunk size distribution~\cite{fastcdc2,chapuis2016throughput,ae2}.
A common approach to achieve a more uniform chunk-size distribution
while minimizing the impact on deduplication ratio is the use of multiple threshold bitmasks.
Multiple thresholds enable dynamic adjustments to the probability of satisfying the chunking condition, increasing or decreasing it as needed.
Most notable are the algorithms in TTTD~\cite{tttd} and \ac{RC}~\cite{rabinPrimary}.
Similar algorithms have later been proposed under the label of \ac{NC}~\cite{fastcdc2,zhou2022ultracdc,wan2022supercdc}.
Furthermore, a popular technique for eliminating pathologically small chunks involves merging adjacent chunks~\cite{tttd,much,sscdc}.
Similarly, pathologically large chunks can be eliminated by rechunking them with a more relaxed chunking condition.

\paragraph*{Deduplication Ratio}
Rechunking large chunks also positively affects the number of redundant chunks that are detected in a dataset.
Different variations of this technique have been formalized,
\eg, within FBC~\cite{fbc} and Bimodal chunking~\cite{bimodal}, among other approaches~\cite{anchorDrivenSubchunk,zhou2013hysteresis,fingerdiff}.


\section{Methodology}\label{sec:setup}
We benchmark eight representative \ac{CDC} algorithms---Rabin, Buzhash, Gear, AE, RAM, MII, PCI, and SeqCDC---in a standardized environment, measuring throughput, chunk-size distribution, and deduplication ratio across target chunk sizes from \qty{512}{B} to \qty{8}{KiB}.
We isolate algorithmic performance from factors not inherent to the chunking algorithms themselves: deployed systems often combine a base algorithm with size thresholds, normalization, caching, skipping, and fingerprinting, which obscure which cut-point rule causes the measured effect.
We therefore tune every algorithm to the same target  chunk sizes on uniformly random data---analytically where possible and by simulation otherwise---freeze those settings, and then evaluate them unchanged on all realistic datasets. We discuss the implications of this choice in \Cref{sec:discussion}.

\begin{table*}[ht]
    \centering
    \caption{Datasets and their Gzip compression ratios.}
    \label{table:datasets}
    \begin{tabularx}{\textwidth}{lrrX}
\toprule
Name & Size & Compress. Ratio & Description \\
\midrule
RAND & \qty{11}{GB} & \qty{100}{\%} & Randomly generated binary file.\\
WEB & \qty{46}{GB} & \qty{88}{\%} & One month of \texttt{tar}-ed daily website snapshots of \texttt{nytimes.com} using \texttt{wget -l 1} and the Internet Archive.\\
DB & \qty{100}{GB} & \qty{31}{\%} & 25 snapshots of a MySQL VM running TPC-C~\cite{leutenegger1993modeling}.\\
VMB & \qty{285}{GB} & \qty{25}{\%} & 33 backups of a nightly Rust build server.\\
CODE & \qty{110}{GB} & \qty{18}{\%} & 207 versions of \texttt{tar}-ed source code distributions for Linux, GDB, GCC, and Emacs.\\
\bottomrule
\end{tabularx}

\end{table*}

We run our benchmarks on a machine with an Intel Xeon Gold 6154 CPU at \qty{3.00}{GHz} and \qty{378}{GB} of memory, running Debian 12.
Throughput and microarchitectural findings describe this system and the evaluated implementations; their magnitude and ranking may change with the processor, compiler, or execution environment.
For identical input bytes, parameters, and deterministic algorithm semantics, chunk boundaries and the resulting deduplication ratios and chunk-size distributions do not depend on execution speed.
Our recommendations concern the evaluated base algorithms; they do not establish rankings for complete deduplication systems or for composites that add thresholds, normalization, caching, or skipping.

To produce a desired mean chunk size, \ac{CDC} algorithms must be tuned through their parameters.
\Cref{sec:analysis} derives parameter rules from each algorithm's stochastic properties, using simulation where analytical derivation is intractable.

\begin{table}[tb]
    \centering
    \caption{Algorithmic parameter settings.}
    \label{table:algoParams}
    \begin{tabularx}{\columnwidth}{Xl}
\toprule
Algorithm & Parameter Settings\\
\midrule
Rabin, Buzhash & $w=32$, $b=\operatorname{round}(\log_2(\mu-w))$ \\
Gear & $b=\operatorname{round}(\log_2\mu)$ \\
AE & Horizons in \Cref{table:aeH} \\
RAM & Chosen in accordance with \Cref{eq:ram-approx}. \\
MII & Chosen in accordance with \Cref{eq:mii}. \\
PCI & Obtained through simulation. \\
SeqCDC & Obtained through simulation. \\
\bottomrule
\end{tabularx}

\end{table}

\subsection{Datasets}\label{sec:setup:datasets}

Because both deduplication ratio and chunk-size distribution strongly depend on the input data, we use multiple realistic datasets and one artificially generated dataset with maximal entropy.
Our dataset selection follows standard practice in \ac{CDC} algorithm evaluations~\cite{fastcdc,ram,bimodal,quickcdc,udayashankar2024seqcdc,udayashankar2025vectorcdc}.
\Cref{table:datasets} lists the datasets and their Gzip compression ratio\footnote{We apply Gzip with the default compression level of 6.}:
$$
\text{Compression ratio}=\frac{\text{Size After Gzip Compression}}{\text{Original Size}}.
$$
Lower values indicate more compressible data.
The datasets are chosen to reflect storage and data-management scenarios in which deduplication is typically deployed.
RAND serves as a high-entropy calibration baseline for parameter tuning and throughput measurements, while the realistic datasets exercise different forms of temporal redundancy.
WEB captures daily website snapshots, where HTML, scripts, metadata, and static assets evolve at different rates.
DB represents repeated backups of a transactional database VM running TPC-C, a setting where database pages, logs, and filesystem layout can all influence redundancy.
VMB captures nightly backups of a build server, combining operating-system state, source trees, package caches, and generated build artifacts.
CODE contains many versions of large source-code distributions and stresses algorithms on compressible text, small files, and archive padding.
For WEB and CODE, we store the file collections as uncompressed TAR archives, while DB and VMB are snapshot-style storage images.
We evaluate these backups and snapshots as streams rather than isolated files.

\subsection{Benchmarking Framework}
\label{sec:setup:benchmark}

To measure both algorithm throughput and the chunks produced, we implement a benchmarking framework in Rust, with a focus on performance and efficiency.
The framework, including the algorithm implementations, is published on GitHub under permissive licenses.\footnote{\url{https://github.com/mrd0ll4r/cdc-bench}}
It simplifies the implementation of new base chunking algorithms and enables their performance to be tested either in isolation or together with high-level techniques such as \ac{NC}, caching, and size limits.
It also excludes overhead from fingerprinting and avoids disk I/O by reading data from a RAMdisk.
The framework consists of a \emph{driver} that reads the input file in large blocks and uses them to drive the selected algorithm.
Each algorithm receives consecutive blocks of file data, identifies boundaries within them, and advances its internal state as it ingests data.
The algorithms closely follow their original descriptions.
All algorithms are integrated into a single benchmark program, compiled with the release profile and tuned for the target machine.
The throughput baseline uses Rust \texttt{1.82.0-nightly}.
The release configuration enables thin link-time optimization, one code-generation unit, and native CPU targeting (\texttt{-C target-cpu=native}).

When evaluating algorithm throughput, we do not fingerprint the generated chunks.
Instead, the benchmark program computes and reports a single aggregate value---the sum of all chunk sizes---to ensure that the compiler cannot eliminate the chunking code through optimization.
We verify that file access does not limit performance by implementing and evaluating fixed-size chunking (FSC) within the same framework.
Sufficient memory is reserved to prevent memory pressure during execution.
Apart from the benchmark process itself, the system remains idle throughout the experiments.
All performance measurements are conducted sequentially to avoid interference between benchmark runs.
For each combination of dataset, algorithm, and target chunk size, the benchmark is executed $n=10$ times.
Additionally, we perform a single warm-up run for each dataset-algorithm combination before collecting measurements, and we capture performance counters using {\ttfamily perf}.
Throughput measurements use RAND, copied to a tmpfs RAMdisk before execution.
The warm-up uses a target chunk size of \qty{512}{B} for each algorithm.
The benchmark runs in a KVM virtual machine without explicit CPU affinity or NUMA placement controls.
CPU frequency is controlled by the host firmware, with power saving disabled.
Where applicable, we normalize measurements by dataset size.
For microarchitectural counters such as \ac{ipc}, instruction count, and branch count, we report mean values; the standard error is at most $0.01$ for all reported metrics and is therefore omitted.
For runtime and throughput, which are commonly skewed in benchmarking experiments~\cite{lemire2023}, we report the median and \ac{iqr}, calculated as $Q_3 - Q_1$.

\section{Parameter Settings and Derivations}\label{sec:analysis}

Every CDC algorithm operates based on specific parameter settings.
We denote the configured target mean chunk size by $\mu$ and the empirical mean by $\bar{c}$.
While $\bar{c}$ depends on both the parameter settings and the characteristics of the respective dataset,
parameter selection can be theoretically grounded with the assumption of uniformly distributed random data.
We therefore use uniformly random data as a common calibration baseline and then evaluate how well the resulting settings transfer to realistic datasets.
This choice separates parameter tuning from dataset-specific behavior and avoids retuning algorithms on the same data used for comparison.

\subsection{\acf{BSW}}
Given random data, a uniformly distributed hash function, and a bitmask of length $b$,
hash-based algorithms find a chunk boundary at each byte with the probability of $2^{-b}$~\cite{chapuis2016throughput,rabinApplications}.
Preliminary experiments confirmed uniform hash-value distributions for Rabin, Buzhash, and Gear. These measurements describe how often hash values occur, but do not test whether successive cut-point decisions are statistically independent.
We set window size $w=32$ for Rabin and Buzhash to align with the implicit window size in Gear.
The implementation selects $b=\operatorname{round}(\log_2(\mu-w))$ for Rabin and Buzhash, and $b=\operatorname{round}(\log_2\mu)$ for Gear, which does not impose the same explicit window-fill delay.
For the five power-of-two targets in this study, both rules yield $b=9,10,11,12,13$, respectively.

\subsection{\acf{AE}}
\label{subsubsec:setup:parameters:ae}

The average chunk size in AE is largely determined by the horizon $h$.
The original authors suggest that $h=\mu/(e-1)$~\cite{ae,ae2}.
This formula has also been implemented in their open-source testing framework Destor~\cite{fu2015design},
as well as another benchmarking tool DedupBench~\cite{liu2023dedupbench} (who corrected the formula following our report~\cite{dedupbenchIssue}).
Following this, the formula was used to produce results for an evaluation of AE in various other works
\cite{ellappan2021dynamic,fastcdc2,jin2023accelerating,ellappan2023smart}.
However, preliminary experimental evaluations on random data show chunk size averages well below the desired target, motivating us to revisit the formula.
Our theoretical analysis then reveals that this formula fails to account for the discrete value range of random variables within $[0,255]$, as well as the conditioning of probabilities on dependent events.
Concretely, the maximum byte value in a horizon of length $h$ equals 255 with probability $1-(255/256)^h$.
For large $h$, the cut predicate therefore approaches the event that the next unseen byte equals 255, which occurs with probability $\frac{1}{256}$.
Based on this reasoning, we propose using the approximation $h\approx\mu -256$ if $\mu\geq\qty{2}{KiB}$ and rely on simulations when target chunk sizes are smaller.
At the \qty{2}{KiB} threshold, the approximate horizon $h=1792$ gives $1-(255/256)^{1792}\approx0.999$.
This probability supports the large-horizon approximation; it is not a bound on its mean-size error.
\Cref{table:aeH} lists the horizons obtained from simulation for the two smaller targets and from $h=\mu-256$ for targets at or above \qty{2}{KiB}.
\Cref{app:ae} provides the detailed stochastic analysis of AE.

\begin{table}[h]
    \centering
    \caption{AE horizon parameterization.}
    \label{table:aeH}
    \begin{tabular}{rr}
\toprule
$\mu$ bytes & $h$ bytes \\
\midrule
512 & 348 \\
1024 & 793 \\
2048 & 1792 \\
4096 & 3840 \\
8192 & 7936 \\
\bottomrule
\end{tabular}
 
\end{table}

\subsection{\acf{RAM}}
\label{sec:analysis:ram}
The authors of RAM do not explicitly state how to tune the horizon $h$
for a desired mean chunk size.
Let $M_h$ be the maximum of $h$ independent, uniformly distributed bytes.
Conditional on $M_h=m$, the number $T$ of subsequent bytes up to and including the first matching byte is geometric with success probability $(256-m)/256$.
Averaging these conditional waiting times gives the model-predicted mean chunk size $\mu(h)$:
\begin{equation}
\begin{split}
\mu(h) = h + \sum_{m=0}^{255}\frac{256}{256-m}
\left[\left(\frac{m+1}{256}\right)^h
-\left(\frac{m}{256}\right)^h\right].
\end{split}
\label{eq:ram}
\end{equation}
This identity assumes an unbounded IID byte stream and counts the matching byte in the chunk.
Minimizing $|\mu(h)-\mu|$ over positive integer horizons provides a theoretical calibration rule for the configured target $\mu$.

The benchmark parameter calculation replaces the random horizon maximum by its expected value, giving the approximation
\begin{equation}
\widetilde\mu(h)=h+\frac{256}{256-\mathrm{E}[M_h]}.
\label{eq:ram-approx}
\end{equation}
This substitution uses a single representative maximum instead of accounting for variation between chunks.
It becomes accurate for large horizons, where $\Pr[M_h=255]=1-(255/256)^h$ approaches one and both means approach $h+256$.
For smaller horizons, however, the maximum varies more, and the waiting time is a nonlinear function of that maximum: evaluating it at the expected maximum underestimates the mean waiting time by convexity.

Solving \Cref{eq:ram-approx} for the target mean and rounding gives the benchmark horizons $h=327,780,1792,3840,7936$ for targets $512,1024,2048,4096,8192$ bytes, respectively.
At the two smallest targets, \Cref{eq:ram} predicts means of \qty{544}{B} and \qty{1030}{B} (rounded to whole bytes), matching the rounded RAND measurements in \Cref{tab:csd_full}.
The unrounded predictions exceed the targets by \qty{6.17}{\%} and \qty{0.57}{\%}, respectively.
The theoretically calibrated horizons are 300 and 774; they are not the settings used for these measurements.
For the three larger targets, both rules select the same integer horizons.
\Cref{app:ram} provides the conditional waiting-time derivation and the detailed approximation-error analysis.

\subsection{\acf{MII}}
The average chunk size in MII depends on the interval parameter $w$.
Our approximation starts from the boundary predicate used by MII:
a chunk boundary is found when a window of $w$ bytes is strictly increasing.
Thus, $w$ counts bytes and a matching window contains $w-1$ increasing adjacent comparisons, consistent with the implementation.
For uniformly random input, each window contains $256^w$ possible byte sequences.
For the predicate to hold, all $w$ bytes must be distinct; for each set of $w$ distinct byte values, exactly one ordering is strictly increasing.
The number of satisfying windows is therefore $\binom{256}{w}$, yielding the approximate boundary probability $\binom{256}{w}\cdot256^{-w}$.
Treating candidate windows as independent gives an expected waiting time equal to the inverse of this probability, plus the window length $w$.
This counting argument captures the factorial growth of the average chunk size with increasing $w$ and the weak control over chunk sizes that follows from it, as also noted by the authors of MII~\cite{mii}.

We use the factor $1.14$ as an empirical correction calibrated on RAND, fixed across target sizes and datasets.
It scales the inverse-probability term, not the additive window length, and is not a theoretical constant.
Dependencies between overlapping windows are a possible explanation for the discrepancy with the independence approximation.
This yields the following approximation:

\begin{equation}
	\mu \approx	1.14\cdot \left(\binom{256}{w}\cdot256^{-w}\right)^{-1} + w
	\label{eq:mii}
\end{equation}

MII presents a special case among the evaluated algorithms,
as $w=5,6,7,8$ bytes give rounded predictions of $147,877,6248,51341$ bytes, respectively.
This makes aligning MII to our range of target chunk sizes impossible.
We use windows $w=6,6,6,7,7$ for targets $512,1024,2048,4096,8192$ bytes, respectively.
For $w=6$ and $w=7$, the rounded RAND means are \qty{887}{B} and \qty{6242}{B} (\Cref{tab:csd_full}).

\subsection{\acf{PCI}}

The average chunk size produced by PCI depends on the setting of its two parameters $w$ (window size) and $\theta$ (threshold).
Specifically, it depends on the ratio $\frac{\theta}{w}$, as well as the absolute value of $w$ (as the window acts as a lower bound on the chunk size).
Because we are not aware of a tractable analytical mapping from $(w,\theta)$ to $\mu$, we obtain suitable values across all target chunk sizes through simulation on RAND.
For each target, we select the parameter pair whose simulated mean chunk size is closest to $\mu$ and use that pair unchanged for all datasets.

\begin{table}[h]
    \centering
    \caption{Parameter settings for PCI.}
    \label{table:pciParams}
    \newcommand{\pciParamRow}[2]{\num{#1} & \num[evaluate-expression]{#1*8} & \num{#2} & \num[evaluate-expression, round-mode=places, round-precision=3]{#2/(8*#1)}}
\begin{tabular}{rrrrr}
\toprule
$\mu$ & $w$ bytes & $8w$ bits & $\theta$ & $\frac{\theta}{8w}$ \\
\midrule
512 & \pciParamRow{58}{253} \\
1024 & \pciParamRow{34}{157} \\
2048 & \pciParamRow{61}{273} \\
4096 & \pciParamRow{39}{183} \\
8192 & \pciParamRow{57}{262} \\
\bottomrule
\end{tabular}
 
\end{table}

\subsection{SeqCDC}
The average chunk size in SeqCDC is the product of three parameters: \emph{SeqLength}, \emph{SkipTrigger}, and \emph{SkipSize}.
For their experimental evaluation,
the authors first generate candidate values via Monte Carlo simulation on randomized data~\cite{udayashankar2024seqcdc}.
Based on those results, they select a combination of values that performs best on one of their datasets.
We note that the choice of those parameters has a large effect on chunking throughput and deduplication ratio.
For example, a higher \emph{SkipSize} is likely to improve chunking throughput at the cost of the deduplication ratio.
Thus, multiple combinations of parameters can yield the same average chunk size and yet differ dramatically on the other metrics.
Since no analytical framework is available, we compile a candidate set through simulation on RAND, analogous to the original paper.
We then select one configuration per target size that matches $\mu$ while avoiding target-specific extremes in \emph{SkipSize}; this keeps the throughput--deduplication tradeoff comparable across target sizes.
The selected configurations are fixed before evaluating realistic datasets.

\begin{table}[h]
    \centering
    \caption{Parameter settings for SeqCDC.}
    \label{table:seqParams}
    \begin{tabular}{rrrrr}
\toprule
$\mu$ & Seq Length & Skip Trigger & Skip Size \\
\midrule
512  & 4 & 65  & 512  \\
  1024 & 5 & 85  & 32   \\
  2048 & 5 & 85  & 256  \\
  4096 & 5 & 120 & 1024 \\
  8192 & 5 & 55  & 1024 \\
\bottomrule
\end{tabular}
 
\end{table}

\section{Computational Efficiency}\label{sec:efficiency}
Much of the discussion and development around chunking algorithms has been fueled by the need to achieve higher throughput
while maintaining good deduplication.
Rabin, one of the oldest \ac{CDC} algorithms, is widely known to be slow,
which has spurred the development of faster \ac{CDC} algorithms.
We evaluate the algorithms in terms of throughput and use microarchitectural details to explain observed performance differences.
Throughput is reported as input processed per execution time.

\subsection{Results}
\label{subsec:perf:results}

We first provide an overview of the achievable throughput using the various algorithms.
To that end, we evaluate each of them on the RAND dataset with a target chunk size of \qty{2}{KiB}.
For better interpretability of the results, we provide them together with relevant microarchitectural counters in \Cref{tab:perf_overview}.
We discuss these results in detail below.


\begin{table}[tb]
	\centering
	\caption{
		Computational performance of chunking algorithms with $\mu=\qty{2}{\kibi\byte}$ on RAND.
		}
\begin{tabularx}{\columnwidth}{Xrrrrrr}
  \toprule
  & \multicolumn{2}{c}{Throughput (MiB/s)} & & & & \\
\cmidrule(lr){2-3}
Alg. & Median & IQR & Instr./B & IPC & Br./B & Br. Mispred. \\ \midrule
FSC & \num{1820} & \num{21} & \num{0.41} & \num{0.24} & \num{0.08} & \num{0.77}\,\% \\ 
  RAM & \num{886} & \num{19} & \num{8.44} & \num{2.39} & \num{1.76} & \num{0.43}\,\% \\ 
  Gear & \num{652} & \num{16} & \num{13.41} & \num{2.78} & \num{1.75} & \num{0.09}\,\% \\ 
  AE & \num{650} & \num{9} & \num{11.41} & \num{2.30} & \num{2.58} & \num{0.16}\,\% \\ 
  SeqCDC & \num{595} & \num{9} & \num{10.59} & \num{1.97} & \num{1.85} & \num{0.28}\,\% \\ 
  Buzhash & \num{311} & \num{6} & \num{29.36} & \num{2.90} & \num{5.07} & \num{0.10}\,\% \\ 
  PCI & \num{242} & \num{3} & \num{35.27} & \num{2.70} & \num{2.56} & \num{0.11}\,\% \\ 
  MII & \num{209} & \num{3} & \num{12.44} & \num{0.82} & \num{2.18} & \num{18.01}\,\% \\ 
  Rabin & \num{195} & \num{2} & \num{34.30} & \num{2.08} & \num{5.06} & \num{0.05}\,\% \\ 
   \bottomrule
\end{tabularx}

\label{tab:perf_overview}
\end{table}

\paragraph*{Hash-Based Algorithms}
The hash-based algorithms Rabin and Buzhash have relatively low throughput in this RAND experiment, while Gear is faster.
We discuss PCI alongside these sliding-window methods because it uses the same BSW scheme, but PCI is hashless: it replaces hashing with a rolling popcount.
Although these algorithms all operate in $\mathcal{O}(1)$ per input byte,
they differ substantially in constant complexity:
Rabin and PCI require more than 30 instructions per input byte, while Buzhash falls just below that,
resulting in the highest throughput of the three at $\qty{311}{MiB/s}$.
Gear uses only $\approx 13$ instructions per byte, leading to a much higher throughput of $\qty{652}{MiB/s}$.
In terms of microarchitectural performance,
all algorithms can be executed relatively well at $>2$ IPC and solid branch predictability.

Among the evaluated hash-based implementations, Gear combines low instruction count with high throughput.
This makes Gear a useful base for further evaluation with high-level techniques such as SIMD or normalized chunking; those combinations are outside this comparison.

\paragraph*{Hashless Algorithms}
The remaining hashless algorithms use extrema or byte comparisons (AE, RAM, SeqCDC, and MII) and generally perform well, with the exception of MII.
Among them, RAM performs best at $\qty{886}{MiB/s}$, markedly better than AE at $\qty{650}{MiB/s}$.
On this system and input, RAM thus achieves its goal of being faster than AE~\cite{ram}.
SeqCDC follows closely at $\qty{595}{MiB/s}$.
This is impressive, considering we elect to implement only the core SeqCDC algorithm,
without the optimization of skipping the start of each chunk.
MII, although algorithmically simple, performs poorly at only $\qty{209}{MiB/s}$.
It executes slightly more instructions per input byte than AE, RAM, and SeqCDC,
but it uses the CPU poorly, reaching only $\qty{0.82}{IPC}$.
A closer look shows that this is due to the poor predictability of the branch comparing the current input byte to the previous one.
We find that \qty{18}{\percent} of branches are mispredicted (``Br. Mispred.''), by far the worst behavior of all examined algorithms.

Among these algorithms (MII, RAM, SeqCDC, AE), MII performs poorly because its branch is difficult to predict.
Both AE and SeqCDC perform well, while RAM has the highest throughput in this experiment.
These throughput rankings apply to the evaluated base implementations on RAND at \qty{2}{KiB} on the measured system.
Finally, if \ac{CDC} is not a concern,
\ac{FSC} unsurprisingly outperforms all CDC contenders at about \qty{1.8}{GiB/s} of processed input.
These results also show that our benchmarking system is not limited by I/O.

\subsection{Key Takeaways}
\label{subsec:perf:conclusions}
Modern CPUs are complex.
It is not always obvious how an algorithm, given in pseudocode, performs on a real system.
We were able to show that various factors of the compiled code and microarchitecture,
in particular constant complexity and branch predictability, play an important role in achieving higher throughput.
Some algorithms (\eg, MII) are algorithmically simple yet difficult for the machine to execute.
Others (\eg, Gear) are both simple and fast to execute.

\begin{figure*}[!ht]
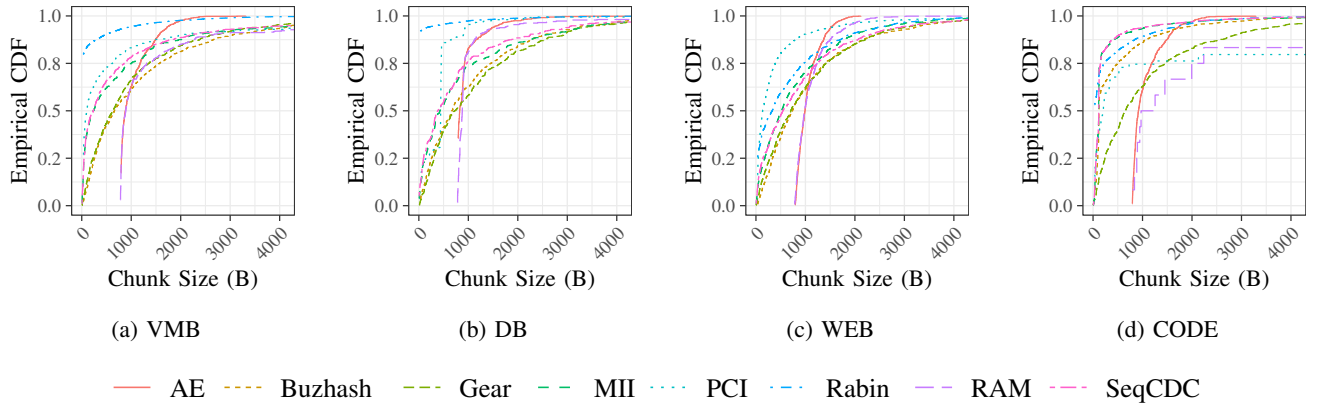

\centering
\subfloat[VMB]{\scalebox{0.8}{\input{fig/csd_vmb.tex}}%
    \label{fig:csd_vmb}}
\hfil
\subfloat[DB]{\scalebox{0.8}{\input{fig/csd_db.tex}}%
    \label{fig:csd_db}}
\hfil
\subfloat[WEB]{\scalebox{0.8}{\input{fig/csd_web.tex}}%
    \label{fig:csd_web}}
\hfil
\subfloat[CODE]{\scalebox{0.8}{\input{fig/csd_code.tex}}%
    \label{fig:csd_code}}
\vspace{-2em}
\begin{tikzpicture}[x=1pt,y=1pt]
\definecolor{fillColor}{RGB}{255,255,255}
\path[use as bounding box,fill=fillColor,fill opacity=0.00] (0,0) rectangle (433.62, 72.27);
\begin{scope}
\path[clip] (  0.00,  0.00) rectangle (433.62, 72.27);
\definecolor{drawColor}{RGB}{248,118,109}

\path[draw=drawColor,line width= 0.6pt,line join=round] ( 13.89, 36.13) -- ( 27.76, 36.13);
\end{scope}
\begin{scope}
\path[clip] (  0.00,  0.00) rectangle (433.62, 72.27);
\definecolor{drawColor}{RGB}{205,150,0}

\path[draw=drawColor,line width= 0.6pt,dash pattern=on 2pt off 2pt ,line join=round] ( 55.62, 36.13) -- ( 69.50, 36.13);
\end{scope}
\begin{scope}
\path[clip] (  0.00,  0.00) rectangle (433.62, 72.27);
\definecolor{drawColor}{RGB}{124,174,0}

\path[draw=drawColor,line width= 0.6pt,dash pattern=on 4pt off 2pt ,line join=round] (123.02, 36.13) -- (136.90, 36.13);
\end{scope}
\begin{scope}
\path[clip] (  0.00,  0.00) rectangle (433.62, 72.27);
\definecolor{drawColor}{RGB}{0,190,103}

\path[draw=drawColor,line width= 0.6pt,dash pattern=on 4pt off 4pt ,line join=round] (173.68, 36.13) -- (187.56, 36.13);
\end{scope}
\begin{scope}
\path[clip] (  0.00,  0.00) rectangle (433.62, 72.27);
\definecolor{drawColor}{RGB}{0,191,196}

\path[draw=drawColor,line width= 0.6pt,dash pattern=on 1pt off 3pt ,line join=round] (215.97, 36.13) -- (229.85, 36.13);
\end{scope}
\begin{scope}
\path[clip] (  0.00,  0.00) rectangle (433.62, 72.27);
\definecolor{drawColor}{RGB}{0,169,255}

\path[draw=drawColor,line width= 0.6pt,dash pattern=on 1pt off 3pt on 4pt off 3pt ,line join=round] (261.05, 36.13) -- (274.92, 36.13);
\end{scope}
\begin{scope}
\path[clip] (  0.00,  0.00) rectangle (433.62, 72.27);
\definecolor{drawColor}{RGB}{199,124,255}

\path[draw=drawColor,line width= 0.6pt,dash pattern=on 7pt off 3pt ,line join=round] (315.61, 36.13) -- (329.49, 36.13);
\end{scope}
\begin{scope}
\path[clip] (  0.00,  0.00) rectangle (433.62, 72.27);
\definecolor{drawColor}{RGB}{255,97,204}

\path[draw=drawColor,line width= 0.6pt,dash pattern=on 2pt off 2pt on 6pt off 2pt ,line join=round] (366.26, 36.13) -- (380.13, 36.13);
\end{scope}
\begin{scope}
\path[clip] (  0.00,  0.00) rectangle (433.62, 72.27);
\definecolor{drawColor}{RGB}{0,0,0}

\node[text=drawColor,anchor=base west,inner sep=0pt, outer sep=0pt, scale=  1.00] at ( 35.00, 32.69) {AE};
\end{scope}
\begin{scope}
\path[clip] (  0.00,  0.00) rectangle (433.62, 72.27);
\definecolor{drawColor}{RGB}{0,0,0}

\node[text=drawColor,anchor=base west,inner sep=0pt, outer sep=0pt, scale=  1.00] at ( 76.73, 32.69) {Buzhash};
\end{scope}
\begin{scope}
\path[clip] (  0.00,  0.00) rectangle (433.62, 72.27);
\definecolor{drawColor}{RGB}{0,0,0}

\node[text=drawColor,anchor=base west,inner sep=0pt, outer sep=0pt, scale=  1.00] at (144.13, 32.69) {Gear};
\end{scope}
\begin{scope}
\path[clip] (  0.00,  0.00) rectangle (433.62, 72.27);
\definecolor{drawColor}{RGB}{0,0,0}

\node[text=drawColor,anchor=base west,inner sep=0pt, outer sep=0pt, scale=  1.00] at (194.79, 32.69) {MII};
\end{scope}
\begin{scope}
\path[clip] (  0.00,  0.00) rectangle (433.62, 72.27);
\definecolor{drawColor}{RGB}{0,0,0}

\node[text=drawColor,anchor=base west,inner sep=0pt, outer sep=0pt, scale=  1.00] at (237.08, 32.69) {PCI};
\end{scope}
\begin{scope}
\path[clip] (  0.00,  0.00) rectangle (433.62, 72.27);
\definecolor{drawColor}{RGB}{0,0,0}

\node[text=drawColor,anchor=base west,inner sep=0pt, outer sep=0pt, scale=  1.00] at (282.16, 32.69) {Rabin};
\end{scope}
\begin{scope}
\path[clip] (  0.00,  0.00) rectangle (433.62, 72.27);
\definecolor{drawColor}{RGB}{0,0,0}

\node[text=drawColor,anchor=base west,inner sep=0pt, outer sep=0pt, scale=  1.00] at (336.72, 32.69) {RAM};
\end{scope}
\begin{scope}
\path[clip] (  0.00,  0.00) rectangle (433.62, 72.27);
\definecolor{drawColor}{RGB}{0,0,0}

\node[text=drawColor,anchor=base west,inner sep=0pt, outer sep=0pt, scale=  1.00] at (387.37, 32.69) {SeqCDC};
\end{scope}
\end{tikzpicture}
\vspace{-2em}
  \caption{Chunk-size distributions for target chunk size $\mu=1024$ truncated at $[0, 4\mu]$.}
  \label{fig:csd_datasets}
\end{figure*}

\section{Chunk Size Distribution}\label{sec:distribution}

Chunk-size distribution affects deduplication ratio, computational efficiency, and storage savings.
Ideally, chunks have empirical mean $\bar{c}$ close to the target $\mu$ without excessive variance: overly large chunks miss deduplication opportunities, while overly small chunks increase processing and metadata overhead.
Because some variation is still needed to recover from boundary shifts, these distributions reveal both tuning accuracy and robustness across datasets with different byte patterns and compressibility.

We characterize each \ac{CDC} algorithm's chunk-size distribution using two metrics:
\begin{enumerate}
	\item The empirical mean $\bar{c}$, whose proximity to $\mu$ indicates tuning accuracy.
	\item The empirical standard deviation $s$, which measures spread around the mean.
\end{enumerate}

\subsection{General Overview}
\label{subsec:csd:general-overview}

We first compare the empirical chunk-size distributions for all algorithms on the realistic datasets at target chunk size $\mu=1024$.
\Cref{fig:csd_datasets} presents these distributions as empirical cumulative distribution functions (CDFs), truncated to the range $[0,4\mu]$ for readability.
For most algorithms and datasets, the empirical CDF is already close to 1 near \qty{4}{KiB}, indicating that almost all chunks remain within four times the target size.
The CODE dataset is the main exception: in \Cref{fig:csd_code}, PCI and RAM show pathological size distributions, with approximately \qty{15}{\%} of chunks exceeding $4\mu$.
More generally, the BSW algorithms produce many below-target chunks together with heavy-tailed distributions, especially on more compressible datasets.
This is because all BSW algorithms follow the same stochastic property:
Each position of the data stream, looked upon independently,
is equally likely to become a chunk cut-point (\cf \Cref{sec:background}).
However, each position's probability also depends on all previous positions not having fulfilled the same matching condition.
Therefore, we observe that these distributions produce the bulk of chunks at a small size,
determined by the window size or other feature of the algorithm,
and then decay, forming a heavy-tailed distribution.
This shifts the mean closer to the target, but
inevitably results in a large chunk-size variance.
AE and RAM, in contrast, impose a lower bound through their horizon, so complete chunk sizes cannot fall far below the target $\mu$.
This lower bound does not constrain the upper tail: AE concentrates chunk sizes near the target, whereas RAM can still produce pathological distributions on realistic datasets.

\begin{table*}[!t]
	\caption{Extensive results on empirical means and standard deviations of produced chunk sizes.}
	\centering
	\definecolor{tempcolor1}{rgb}{1,0.96,0.94} 
\definecolor{tempcolor2}{rgb}{0.99,0.88,0.82} 
\definecolor{tempcolor3}{rgb}{0.79,0.09,0.11} 
\definecolor{tempcolor4}{rgb}{0.98,0.41,0.29} 
\definecolor{tempcolor5}{rgb}{0.98,0.57,0.45} 
\definecolor{tempcolor6}{rgb}{0.4,0,0.05} 
\definecolor{tempcolor7}{rgb}{0.93,0.23,0.17} 
\definecolor{tempcolor8}{rgb}{0.98,0.73,0.63} 
\definecolor{tempcolor9}{rgb}{0.64,0.06,0.08} 
\definecolor{black}{rgb}{0.0, 0.0, 0.0} 
\color{black}
\begin{tabularx}{\textwidth}{Xcrrrrrrrrrrrrrr}
\multirow{2}{*}{Algorithm}&&\multicolumn{2}{c}{512\,B}&&\multicolumn{2}{c}{1024\,B}&&\multicolumn{2}{c}{2048\,B}&&\multicolumn{2}{c}{4096\,B}&&\multicolumn{2}{c}{8192\,B}\\ 
\cline{3-4}\cline{6-7}\cline{9-10}\cline{12-13}\cline{15-16}
&&$\bar{c}$&$s$&&$\bar{c}$&$s$&&$\bar{c}$&$s$&&$\bar{c}$&$s$&&$\bar{c}$&$s$\\ 
\hline\noalign{\smallskip}

\textbf{RANDOM}&&&&&&&&&&&&&&&\\ 
Rabin&&\cellcolor{tempcolor1}542&\cellcolor{tempcolor5}511&&\cellcolor{tempcolor1}1054&\cellcolor{tempcolor5}1023&&\cellcolor{tempcolor1}2078&\cellcolor{tempcolor4}\color{white}2046&&\cellcolor{tempcolor1}4123&\cellcolor{tempcolor4}\color{white}4096&&\cellcolor{tempcolor1}8213&\cellcolor{tempcolor4}\color{white}8185\\ 
Buzhash&&\cellcolor{tempcolor2}544&\cellcolor{tempcolor5}512&&\cellcolor{tempcolor1}1057&\cellcolor{tempcolor5}1026&&\cellcolor{tempcolor1}2084&\cellcolor{tempcolor4}\color{white}2053&&\cellcolor{tempcolor1}4137&\cellcolor{tempcolor4}\color{white}4106&&\cellcolor{tempcolor1}8240&\cellcolor{tempcolor4}\color{white}8208\\ 
Gear&&\cellcolor{tempcolor1}511&\cellcolor{tempcolor5}510&&\cellcolor{tempcolor1}1024&\cellcolor{tempcolor4}\color{white}1022&&\cellcolor{tempcolor1}2047&\cellcolor{tempcolor4}\color{white}2045&&\cellcolor{tempcolor1}4094&\cellcolor{tempcolor4}\color{white}4088&&\cellcolor{tempcolor1}8197&\cellcolor{tempcolor4}\color{white}8181\\ 
AE&&\cellcolor{tempcolor1}512&\cellcolor{tempcolor1}136&&\cellcolor{tempcolor1}1024&\cellcolor{tempcolor1}209&&\cellcolor{tempcolor1}2048&\cellcolor{tempcolor1}252&&\cellcolor{tempcolor1}4095&\cellcolor{tempcolor1}255&&\cellcolor{tempcolor1}8191&\cellcolor{tempcolor1}255\\ 
RAM&&\cellcolor{tempcolor2}544&\cellcolor{tempcolor2}234&&\cellcolor{tempcolor1}1030&\cellcolor{tempcolor1}252&&\cellcolor{tempcolor1}2048&\cellcolor{tempcolor1}255&&\cellcolor{tempcolor1}4096&\cellcolor{tempcolor1}255&&\cellcolor{tempcolor1}8193&\cellcolor{tempcolor1}255\\ 
PCI&&\cellcolor{tempcolor1}506&\cellcolor{tempcolor5}465&&\cellcolor{tempcolor1}1017&\cellcolor{tempcolor5}990&&\cellcolor{tempcolor1}2085&\cellcolor{tempcolor4}\color{white}2038&&\cellcolor{tempcolor1}4290&\cellcolor{tempcolor4}\color{white}4257&&\cellcolor{tempcolor1}8693&\cellcolor{tempcolor4}\color{white}8637\\ 
MII&&\cellcolor{tempcolor3}\color{white}887&\cellcolor{tempcolor5}881&&\cellcolor{tempcolor2}887&\cellcolor{tempcolor4}\color{white}881&&\cellcolor{tempcolor7}\color{white}887&\cellcolor{tempcolor4}\color{white}881&&\cellcolor{tempcolor4}\color{white}6242&\cellcolor{tempcolor4}\color{white}6238&&\cellcolor{tempcolor8}6242&\cellcolor{tempcolor4}\color{white}6238\\ 
SeqCDC&&\cellcolor{tempcolor1}506&\cellcolor{tempcolor7}\color{white}698&&\cellcolor{tempcolor1}1026&\cellcolor{tempcolor4}\color{white}1036&&\cellcolor{tempcolor1}2069&\cellcolor{tempcolor4}\color{white}2185&&\cellcolor{tempcolor1}4110&\cellcolor{tempcolor4}\color{white}4578&&\cellcolor{tempcolor1}8544&\cellcolor{tempcolor4}\color{white}9024\\ 
\hline\noalign{\smallskip}
\textbf{CODE}&&&&&&&&&&&&&&&\\ 
Rabin&&\cellcolor{tempcolor4}\color{white}278&\cellcolor{tempcolor6}\color{white}1201&&\cellcolor{tempcolor7}\color{white}364&\cellcolor{tempcolor6}\color{white}1507&&\cellcolor{tempcolor3}\color{white}427&\cellcolor{tempcolor6}\color{white}1869&&\cellcolor{tempcolor9}\color{white}478&\cellcolor{tempcolor6}\color{white}2471&&\cellcolor{tempcolor9}\color{white}517&\cellcolor{tempcolor6}\color{white}3475\\ 
Buzhash&&\cellcolor{tempcolor5}351&\cellcolor{tempcolor6}\color{white}1317&&\cellcolor{tempcolor4}\color{white}476&\cellcolor{tempcolor6}\color{white}1668&&\cellcolor{tempcolor3}\color{white}628&\cellcolor{tempcolor6}\color{white}2395&&\cellcolor{tempcolor9}\color{white}720&\cellcolor{tempcolor6}\color{white}3234&&\cellcolor{tempcolor9}\color{white}775&\cellcolor{tempcolor6}\color{white}4437\\ 
Gear&&\cellcolor{tempcolor2}572&\cellcolor{tempcolor6}\color{white}1718&&\cellcolor{tempcolor2}1165&\cellcolor{tempcolor6}\color{white}2706&&\cellcolor{tempcolor2}2305&\cellcolor{tempcolor6}\color{white}4496&&\cellcolor{tempcolor2}4706&\cellcolor{tempcolor9}\color{white}8450&&\cellcolor{tempcolor2}9358&\cellcolor{tempcolor9}\color{white}15605\\ 
AE&&\cellcolor{tempcolor1}484&\cellcolor{tempcolor1}147&&\cellcolor{tempcolor1}1046&\cellcolor{tempcolor1}313&&\cellcolor{tempcolor2}2228&\cellcolor{tempcolor2}588&&\cellcolor{tempcolor2}4594&\cellcolor{tempcolor2}1135&&\cellcolor{tempcolor2}9332&\cellcolor{tempcolor2}2266\\ 
RAM&&\cellcolor{tempcolor6}\color{white}32663&\cellcolor{tempcolor6}\color{white}3620803&&\cellcolor{tempcolor6}\color{white}61202&\cellcolor{tempcolor6}\color{white}4955982&&\cellcolor{tempcolor6}\color{white}120909&\cellcolor{tempcolor6}\color{white}6989936&&\cellcolor{tempcolor6}\color{white}219247&\cellcolor{tempcolor6}\color{white}9411473&&\cellcolor{tempcolor6}\color{white}387964&\cellcolor{tempcolor6}\color{white}12516487\\ 
PCI&&\cellcolor{tempcolor6}\color{white}6609&\cellcolor{tempcolor6}\color{white}116719&&\cellcolor{tempcolor6}\color{white}10722&\cellcolor{tempcolor6}\color{white}259899&&\cellcolor{tempcolor6}\color{white}22598&\cellcolor{tempcolor6}\color{white}948426&&\cellcolor{tempcolor6}\color{white}30252&\cellcolor{tempcolor6}\color{white}883873&&\cellcolor{tempcolor6}\color{white}44879&\cellcolor{tempcolor6}\color{white}1627986\\ 
MII&&\cellcolor{tempcolor5}333&\cellcolor{tempcolor6}\color{white}2977&&\cellcolor{tempcolor7}\color{white}333&\cellcolor{tempcolor6}\color{white}2977&&\cellcolor{tempcolor9}\color{white}333&\cellcolor{tempcolor6}\color{white}2977&&\cellcolor{tempcolor9}\color{white}481&\cellcolor{tempcolor6}\color{white}6299&&\cellcolor{tempcolor6}\color{white}481&\cellcolor{tempcolor6}\color{white}6299\\ 
SeqCDC&&\cellcolor{tempcolor7}\color{white}191&\cellcolor{tempcolor6}\color{white}1521&&\cellcolor{tempcolor3}\color{white}318&\cellcolor{tempcolor6}\color{white}2794&&\cellcolor{tempcolor9}\color{white}373&\cellcolor{tempcolor6}\color{white}3199&&\cellcolor{tempcolor9}\color{white}420&\cellcolor{tempcolor6}\color{white}3731&&\cellcolor{tempcolor9}\color{white}561&\cellcolor{tempcolor6}\color{white}4782\\ 
\hline\noalign{\smallskip}
\textbf{WEB}&&&&&&&&&&&&&&&\\ 
Rabin&&\cellcolor{tempcolor2}444&\cellcolor{tempcolor4}\color{white}525&&\cellcolor{tempcolor8}717&\cellcolor{tempcolor7}\color{white}1002&&\cellcolor{tempcolor4}\color{white}1071&\cellcolor{tempcolor9}\color{white}1868&&\cellcolor{tempcolor7}\color{white}1422&\cellcolor{tempcolor6}\color{white}3255&&\cellcolor{tempcolor3}\color{white}1702&\cellcolor{tempcolor6}\color{white}5231\\ 
Buzhash&&\cellcolor{tempcolor2}551&\cellcolor{tempcolor5}542&&\cellcolor{tempcolor1}1058&\cellcolor{tempcolor4}\color{white}1082&&\cellcolor{tempcolor1}2090&\cellcolor{tempcolor4}\color{white}2195&&\cellcolor{tempcolor1}4009&\cellcolor{tempcolor4}\color{white}4252&&\cellcolor{tempcolor1}7740&\cellcolor{tempcolor7}\color{white}9274\\ 
Gear&&\cellcolor{tempcolor1}523&\cellcolor{tempcolor4}\color{white}561&&\cellcolor{tempcolor1}1039&\cellcolor{tempcolor4}\color{white}1103&&\cellcolor{tempcolor1}2076&\cellcolor{tempcolor4}\color{white}2173&&\cellcolor{tempcolor1}4189&\cellcolor{tempcolor4}\color{white}4491&&\cellcolor{tempcolor1}8498&\cellcolor{tempcolor4}\color{white}9534\\ 
AE&&\cellcolor{tempcolor1}517&\cellcolor{tempcolor1}147&&\cellcolor{tempcolor1}1048&\cellcolor{tempcolor1}245&&\cellcolor{tempcolor1}2124&\cellcolor{tempcolor1}374&&\cellcolor{tempcolor1}4224&\cellcolor{tempcolor1}536&&\cellcolor{tempcolor1}8406&\cellcolor{tempcolor1}904\\ 
RAM&&\cellcolor{tempcolor5}672&\cellcolor{tempcolor6}\color{white}53567&&\cellcolor{tempcolor8}1256&\cellcolor{tempcolor6}\color{white}73244&&\cellcolor{tempcolor8}2435&\cellcolor{tempcolor6}\color{white}101956&&\cellcolor{tempcolor2}4764&\cellcolor{tempcolor6}\color{white}142598&&\cellcolor{tempcolor2}9414&\cellcolor{tempcolor6}\color{white}200436\\ 
PCI&&\cellcolor{tempcolor8}357&\cellcolor{tempcolor6}\color{white}2007&&\cellcolor{tempcolor4}\color{white}491&\cellcolor{tempcolor6}\color{white}3523&&\cellcolor{tempcolor7}\color{white}739&\cellcolor{tempcolor6}\color{white}4617&&\cellcolor{tempcolor3}\color{white}994&\cellcolor{tempcolor6}\color{white}7760&&\cellcolor{tempcolor9}\color{white}1323&\cellcolor{tempcolor6}\color{white}15125\\ 
MII&&\cellcolor{tempcolor7}\color{white}841&\cellcolor{tempcolor3}\color{white}1304&&\cellcolor{tempcolor2}841&\cellcolor{tempcolor3}\color{white}1304&&\cellcolor{tempcolor7}\color{white}841&\cellcolor{tempcolor3}\color{white}1304&&\cellcolor{tempcolor8}4970&\cellcolor{tempcolor9}\color{white}8797&&\cellcolor{tempcolor5}4970&\cellcolor{tempcolor9}\color{white}8797\\ 
SeqCDC&&\cellcolor{tempcolor2}424&\cellcolor{tempcolor9}\color{white}744&&\cellcolor{tempcolor2}943&\cellcolor{tempcolor3}\color{white}1424&&\cellcolor{tempcolor2}1805&\cellcolor{tempcolor3}\color{white}2665&&\cellcolor{tempcolor2}3377&\cellcolor{tempcolor3}\color{white}5158&&\cellcolor{tempcolor8}6572&\cellcolor{tempcolor3}\color{white}9467\\ 
\hline\noalign{\smallskip}
\textbf{DB}&&&&&&&&&&&&&&&\\ 
Rabin&&\cellcolor{tempcolor9}\color{white}95&\cellcolor{tempcolor6}\color{white}277&&\cellcolor{tempcolor9}\color{white}101&\cellcolor{tempcolor6}\color{white}402&&\cellcolor{tempcolor6}\color{white}104&\cellcolor{tempcolor6}\color{white}571&&\cellcolor{tempcolor6}\color{white}106&\cellcolor{tempcolor6}\color{white}796&&\cellcolor{tempcolor6}\color{white}107&\cellcolor{tempcolor6}\color{white}1081\\ 
Buzhash&&\cellcolor{tempcolor7}\color{white}823&\cellcolor{tempcolor6}\color{white}47860&&\cellcolor{tempcolor4}\color{white}1579&\cellcolor{tempcolor6}\color{white}66862&&\cellcolor{tempcolor4}\color{white}3050&\cellcolor{tempcolor6}\color{white}93309&&\cellcolor{tempcolor5}5670&\cellcolor{tempcolor6}\color{white}127677&&\cellcolor{tempcolor8}10108&\cellcolor{tempcolor6}\color{white}172277\\ 
Gear&&\cellcolor{tempcolor1}524&\cellcolor{tempcolor6}\color{white}38456&&\cellcolor{tempcolor4}\color{white}1531&\cellcolor{tempcolor6}\color{white}65776&&\cellcolor{tempcolor4}\color{white}3006&\cellcolor{tempcolor6}\color{white}92557&&\cellcolor{tempcolor4}\color{white}6162&\cellcolor{tempcolor6}\color{white}134051&&\cellcolor{tempcolor4}\color{white}12129&\cellcolor{tempcolor6}\color{white}190428\\ 
AE&&\cellcolor{tempcolor2}448&\cellcolor{tempcolor2}220&&\cellcolor{tempcolor2}949&\cellcolor{tempcolor2}339&&\cellcolor{tempcolor1}2079&\cellcolor{tempcolor2}618&&\cellcolor{tempcolor1}4328&\cellcolor{tempcolor2}1107&&\cellcolor{tempcolor2}8780&\cellcolor{tempcolor2}2180\\ 
RAM&&\cellcolor{tempcolor3}\color{white}886&\cellcolor{tempcolor6}\color{white}58702&&\cellcolor{tempcolor3}\color{white}1741&\cellcolor{tempcolor6}\color{white}82555&&\cellcolor{tempcolor3}\color{white}3601&\cellcolor{tempcolor6}\color{white}117545&&\cellcolor{tempcolor3}\color{white}7147&\cellcolor{tempcolor6}\color{white}166900&&\cellcolor{tempcolor3}\color{white}13830&\cellcolor{tempcolor6}\color{white}231920\\ 
PCI&&\cellcolor{tempcolor6}\color{white}4481&\cellcolor{tempcolor6}\color{white}376356&&\cellcolor{tempcolor9}\color{white}1968&\cellcolor{tempcolor6}\color{white}264093&&\cellcolor{tempcolor6}\color{white}8546&\cellcolor{tempcolor6}\color{white}1670258&&\cellcolor{tempcolor9}\color{white}7853&\cellcolor{tempcolor6}\color{white}1121784&&\cellcolor{tempcolor5}11720&\cellcolor{tempcolor6}\color{white}2037800\\ 
MII&&\cellcolor{tempcolor6}\color{white}1490&\cellcolor{tempcolor6}\color{white}69072&&\cellcolor{tempcolor4}\color{white}1490&\cellcolor{tempcolor6}\color{white}69072&&\cellcolor{tempcolor8}1490&\cellcolor{tempcolor6}\color{white}69072&&\cellcolor{tempcolor6}\color{white}8724&\cellcolor{tempcolor6}\color{white}185368&&\cellcolor{tempcolor2}8724&\cellcolor{tempcolor6}\color{white}185368\\ 
SeqCDC&&\cellcolor{tempcolor8}631&\cellcolor{tempcolor6}\color{white}43950&&\cellcolor{tempcolor4}\color{white}1591&\cellcolor{tempcolor6}\color{white}71362&&\cellcolor{tempcolor5}2785&\cellcolor{tempcolor6}\color{white}95023&&\cellcolor{tempcolor8}5039&\cellcolor{tempcolor6}\color{white}128883&&\cellcolor{tempcolor2}9401&\cellcolor{tempcolor6}\color{white}180942\\ 
\hline\noalign{\smallskip}
\textbf{VMB}&&&&&&&&&&&&&&&\\ 
Rabin&&\cellcolor{tempcolor7}\color{white}170&\cellcolor{tempcolor6}\color{white}653&&\cellcolor{tempcolor3}\color{white}194&\cellcolor{tempcolor6}\color{white}843&&\cellcolor{tempcolor9}\color{white}212&\cellcolor{tempcolor6}\color{white}1143&&\cellcolor{tempcolor6}\color{white}222&\cellcolor{tempcolor6}\color{white}1549&&\cellcolor{tempcolor6}\color{white}227&\cellcolor{tempcolor6}\color{white}2144\\ 
Buzhash&&\cellcolor{tempcolor5}676&\cellcolor{tempcolor6}\color{white}16728&&\cellcolor{tempcolor8}1312&\cellcolor{tempcolor6}\color{white}23348&&\cellcolor{tempcolor8}2455&\cellcolor{tempcolor6}\color{white}32106&&\cellcolor{tempcolor8}4949&\cellcolor{tempcolor6}\color{white}46735&&\cellcolor{tempcolor2}9063&\cellcolor{tempcolor6}\color{white}66011\\ 
Gear&&\cellcolor{tempcolor2}602&\cellcolor{tempcolor6}\color{white}15855&&\cellcolor{tempcolor8}1222&\cellcolor{tempcolor6}\color{white}23730&&\cellcolor{tempcolor8}2435&\cellcolor{tempcolor6}\color{white}33795&&\cellcolor{tempcolor8}4867&\cellcolor{tempcolor6}\color{white}48583&&\cellcolor{tempcolor2}9184&\cellcolor{tempcolor6}\color{white}69614\\ 
AE&&\cellcolor{tempcolor2}471&\cellcolor{tempcolor1}163&&\cellcolor{tempcolor1}1047&\cellcolor{tempcolor2}352&&\cellcolor{tempcolor2}2297&\cellcolor{tempcolor2}706&&\cellcolor{tempcolor2}4807&\cellcolor{tempcolor2}1438&&\cellcolor{tempcolor2}9542&\cellcolor{tempcolor2}2664\\ 
RAM&&\cellcolor{tempcolor6}\color{white}1298&\cellcolor{tempcolor6}\color{white}112214&&\cellcolor{tempcolor6}\color{white}2603&\cellcolor{tempcolor6}\color{white}160204&&\cellcolor{tempcolor6}\color{white}5141&\cellcolor{tempcolor6}\color{white}256212&&\cellcolor{tempcolor6}\color{white}9595&\cellcolor{tempcolor6}\color{white}351555&&\cellcolor{tempcolor6}\color{white}17255&\cellcolor{tempcolor6}\color{white}415735\\ 
PCI&&\cellcolor{tempcolor6}\color{white}2150&\cellcolor{tempcolor6}\color{white}93943&&\cellcolor{tempcolor6}\color{white}3166&\cellcolor{tempcolor6}\color{white}127712&&\cellcolor{tempcolor6}\color{white}5365&\cellcolor{tempcolor6}\color{white}167882&&\cellcolor{tempcolor4}\color{white}6063&\cellcolor{tempcolor6}\color{white}197997&&\cellcolor{tempcolor2}8925&\cellcolor{tempcolor6}\color{white}253655\\ 
MII&&\cellcolor{tempcolor6}\color{white}1619&\cellcolor{tempcolor6}\color{white}34972&&\cellcolor{tempcolor7}\color{white}1619&\cellcolor{tempcolor6}\color{white}34972&&\cellcolor{tempcolor8}1619&\cellcolor{tempcolor6}\color{white}34972&&\cellcolor{tempcolor6}\color{white}10623&\cellcolor{tempcolor6}\color{white}122867&&\cellcolor{tempcolor8}10623&\cellcolor{tempcolor6}\color{white}122867\\ 
SeqCDC&&\cellcolor{tempcolor1}503&\cellcolor{tempcolor6}\color{white}17570&&\cellcolor{tempcolor4}\color{white}1538&\cellcolor{tempcolor6}\color{white}34052&&\cellcolor{tempcolor2}2391&\cellcolor{tempcolor6}\color{white}42843&&\cellcolor{tempcolor2}3717&\cellcolor{tempcolor6}\color{white}57142&&\cellcolor{tempcolor8}6565&\cellcolor{tempcolor6}\color{white}77491\\ 
\hline
\end{tabularx}

	\label{tab:csd_full}
\end{table*}

\subsection{Mean and Variance}

We report the mean chunk size $\bar{c}$ and standard deviation $s$ across all datasets and algorithms. 
A detailed quantitative table is provided in \Cref{tab:csd_full}.
Cell colors follow a continuous scale, with lighter shades indicating values closer to the optimum, which for $\bar{c}$ corresponds to $\mu$.
The color range reaches its maximum where the mean deviates $\pm\qty{100}{\%}$ from the target $\mu$.
For the standard deviation $s$,
the color range is relative to the empirical mean $\bar{c}$, in a range $[0, 2\bar{c}]$.
Values are aggregated across all evaluated target chunk sizes.
Recall from \Cref{sec:setup} that we examine algorithms without constraining chunk sizes through minimum or maximum thresholds.


We make the following observations:
Firstly, RAND provides an empirical calibration check for AE and MII.
For RAM, the measured means agree closely with the expectation in \Cref{eq:ram} at the benchmark horizons.
In particular, the configuration at a \qty{512}{B} target produces a mean of \qty{544}{B}, consistent with the parameter-selection approximation error quantified in \Cref{sec:analysis:ram}.
In general, all algorithms perform predictably within acceptable bounds on artificial pseudorandom data.
On real-world datasets, results are much more varied.
For many algorithms, the more compressible datasets tend to yield greater chunk-size variance.
The compression ratios in \Cref{table:datasets} summarize dataset compressibility.
Often, with hash-based algorithms, we furthermore observe an increase in variance when target chunk sizes are higher.

\paragraph*{Hash-Based Algorithms}
Different hash functions within hash-based algorithms yield slightly different chunk-size distributions.
While Gear tends to maintain mean sizes closer to the target,
all hash-based methods struggle with chunk-size variance.
In particular, Rabin and Buzhash perform poorly on the more compressible datasets DB, CODE, and VMB.
One hypothesis is that Gear's chunking rule interacts more favorably with the byte patterns in these datasets; isolating that mechanism requires further experiments.


%

\paragraph*{Hashless Algorithms}
The lowest levels of variance overall are produced by AE.
RAM achieves similarly low variance only on the RAND dataset.
On the realistic datasets, RAM often fails to find chunk boundaries frequently enough, resulting in pathological chunk-size distributions.
The authors of RAM warn against poor performance on files containing low-entropy data.
In their paper~\cite{ram}, the experiments yield very similar results in comparison with AE.
We find pathological performance on all realistic datasets, including WEB, whose Gzip compression ratio is comparatively high.
SeqCDC also produces means close to the target on all datasets except CODE, indicating a potential sensitivity to CODE's byte patterns.

\paragraph*{CDC for Data Synchronization}
PCI exhibits a significant deviation from the target chunk size and high chunk-size variance across most settings.
Its predecessor, MII, although restrictive in the targets it can tune to, performs slightly better in this regard.
However, since chunk-size variance is not a critical concern in the application for data synchronization, this may not be a substantial drawback.

\subsection{Key Takeaways}
\label{subsec:csd:takeaways}

The results on synthetic and real-world datasets lead to the following conclusions about CDC chunk-size distributions:

\paragraph*{Chunk-Size Variance}
Almost all algorithms produce many small chunks near an algorithm-defined minimum, followed by a heavy tail that moves the mean closer to the target.
Rolling hash-based algorithms use relatively small windows and generally exhibit longer tails and greater variance.
AE and RAM both impose a horizon-based lower bound near the configured target, reducing the opportunity for very small chunks. Only AE consistently achieves low variance in our realistic workloads; RAM's long upper tails remain pathological despite that lower bound.

\paragraph*{Performance on Real-World Datasets}
We find that all algorithms can be tuned to perform well on pseudorandom data,
which often forms the basis for their theoretical analyses.
On real-world data, however, we see mixed results, depending on the algorithm and dataset combination.
Among the evaluated base chunking algorithms, AE is the most consistent across datasets, with Gear and SeqCDC also performing well on several metrics.


\section{Deduplication Ratio}\label{sec:dedup}

\begin{figure*}[ht]
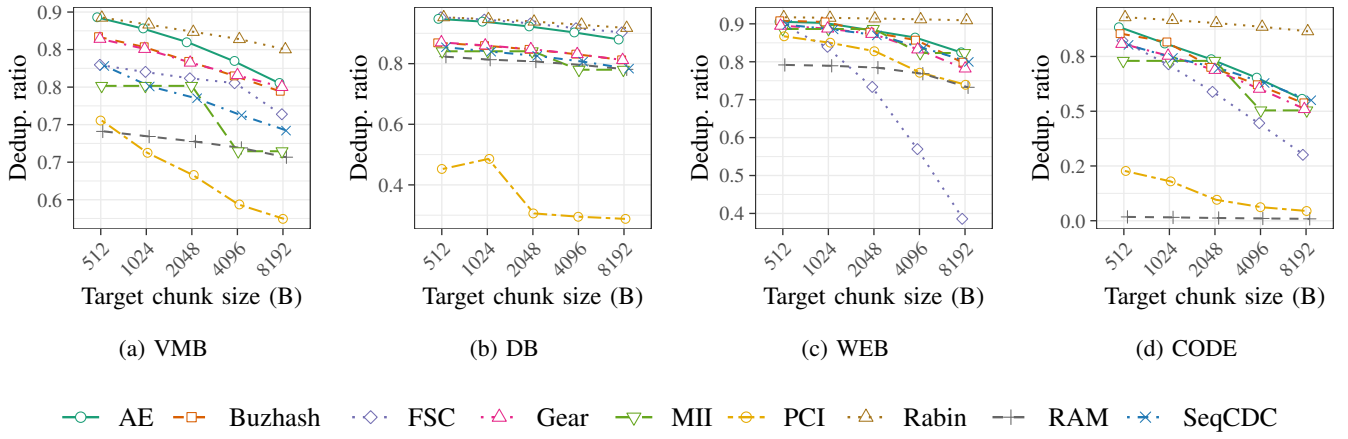

\centering
\subfloat[VMB]{\scalebox{0.85}{\input{fig/dedup_overview_vmb}}%
    \label{fig:dedup_overview:vmb}}
\hfil
\subfloat[DB]{\scalebox{0.85}{\input{fig/dedup_overview_db}}%
	\label{fig:dedup_overview:db}}
\hfil
\subfloat[WEB]{\scalebox{0.85}{\input{fig/dedup_overview_web}}%
    \label{fig:dedup_overview:web}}
\hfil
\subfloat[CODE]{\scalebox{0.85}{\input{fig/dedup_overview_code}}%
    \label{fig:dedup_overview:code}}
 
\vspace{-1.6em}
\begin{tikzpicture}[x=1pt,y=1pt]
\definecolor{fillColor}{RGB}{255,255,255}
\path[use as bounding box,fill=fillColor,fill opacity=0.00] (0,0) rectangle (505.89, 72.27);
\begin{scope}
\path[clip] (  0.00,  0.00) rectangle (505.89, 72.27);
\definecolor{drawColor}{RGB}{27,158,119}

\path[draw=drawColor,line width= 0.9pt,line join=round] ( 25.81, 36.13) -- ( 39.69, 36.13);
\end{scope}
\begin{scope}
\path[clip] (  0.00,  0.00) rectangle (505.89, 72.27);
\definecolor{drawColor}{RGB}{27,158,119}
\definecolor{fillColor}{RGB}{255,255,255}

\path[draw=drawColor,line width= 0.4pt,line join=round,line cap=round,fill=fillColor] ( 32.75, 36.13) circle (  1.96);
\end{scope}
\begin{scope}
\path[clip] (  0.00,  0.00) rectangle (505.89, 72.27);
\definecolor{drawColor}{RGB}{217,95,2}

\path[draw=drawColor,line width= 0.9pt,dash pattern=on 4pt off 4pt ,line join=round] ( 67.55, 36.13) -- ( 81.42, 36.13);
\end{scope}
\begin{scope}
\path[clip] (  0.00,  0.00) rectangle (505.89, 72.27);
\definecolor{drawColor}{RGB}{217,95,2}
\definecolor{fillColor}{RGB}{255,255,255}

\path[draw=drawColor,line width= 0.4pt,line join=round,line cap=round,fill=fillColor] ( 72.75, 34.40) rectangle ( 76.23, 37.87);
\end{scope}
\begin{scope}
\path[clip] (  0.00,  0.00) rectangle (505.89, 72.27);
\definecolor{drawColor}{RGB}{117,112,179}

\path[draw=drawColor,line width= 0.9pt,dash pattern=on 1pt off 3pt ,line join=round] (134.95, 36.13) -- (148.83, 36.13);
\end{scope}
\begin{scope}
\path[clip] (  0.00,  0.00) rectangle (505.89, 72.27);
\definecolor{drawColor}{RGB}{117,112,179}
\definecolor{fillColor}{RGB}{255,255,255}

\path[draw=drawColor,line width= 0.4pt,line join=round,line cap=round,fill=fillColor] (141.89, 33.68) --
	(144.35, 36.13) --
	(141.89, 38.59) --
	(139.43, 36.13) --
	cycle;
\end{scope}
\begin{scope}
\path[clip] (  0.00,  0.00) rectangle (505.89, 72.27);
\definecolor{drawColor}{RGB}{231,41,138}

\path[draw=drawColor,line width= 0.9pt,dash pattern=on 1pt off 3pt on 4pt off 3pt ,line join=round] (183.37, 36.13) -- (197.25, 36.13);
\end{scope}
\begin{scope}
\path[clip] (  0.00,  0.00) rectangle (505.89, 72.27);
\definecolor{drawColor}{RGB}{231,41,138}
\definecolor{fillColor}{RGB}{255,255,255}

\path[draw=drawColor,line width= 0.4pt,line join=round,line cap=round,fill=fillColor] (190.31, 39.19) --
	(192.95, 34.61) --
	(187.66, 34.61) --
	cycle;
\end{scope}
\begin{scope}
\path[clip] (  0.00,  0.00) rectangle (505.89, 72.27);
\definecolor{drawColor}{RGB}{102,166,30}

\path[draw=drawColor,line width= 0.9pt,dash pattern=on 7pt off 3pt ,line join=round] (234.03, 36.13) -- (247.90, 36.13);
\end{scope}
\begin{scope}
\path[clip] (  0.00,  0.00) rectangle (505.89, 72.27);
\definecolor{drawColor}{RGB}{102,166,30}
\definecolor{fillColor}{RGB}{255,255,255}

\path[draw=drawColor,line width= 0.4pt,line join=round,line cap=round,fill=fillColor] (240.97, 33.08) --
	(243.61, 37.66) --
	(238.32, 37.66) --
	cycle;
\end{scope}
\begin{scope}
\path[clip] (  0.00,  0.00) rectangle (505.89, 72.27);
\definecolor{drawColor}{RGB}{230,171,2}

\path[draw=drawColor,line width= 0.9pt,dash pattern=on 2pt off 2pt on 6pt off 2pt ,line join=round] (276.31, 36.13) -- (290.19, 36.13);
\end{scope}
\begin{scope}
\path[clip] (  0.00,  0.00) rectangle (505.89, 72.27);
\definecolor{drawColor}{RGB}{230,171,2}

\path[draw=drawColor,line width= 0.4pt,line join=round,line cap=round] (283.25, 36.13) circle (  1.96);
\end{scope}
\begin{scope}
\path[clip] (  0.00,  0.00) rectangle (505.89, 72.27);
\definecolor{drawColor}{RGB}{166,118,29}

\path[draw=drawColor,line width= 0.9pt,dash pattern=on 1pt off 3pt ,line join=round] (321.39, 36.13) -- (335.27, 36.13);
\end{scope}
\begin{scope}
\path[clip] (  0.00,  0.00) rectangle (505.89, 72.27);
\definecolor{drawColor}{RGB}{166,118,29}

\path[draw=drawColor,line width= 0.4pt,line join=round,line cap=round] (328.33, 39.19) --
	(330.97, 34.61) --
	(325.69, 34.61) --
	cycle;
\end{scope}
\begin{scope}
\path[clip] (  0.00,  0.00) rectangle (505.89, 72.27);
\definecolor{drawColor}{gray}{0.40}

\path[draw=drawColor,line width= 0.9pt,dash pattern=on 4pt off 4pt ,line join=round] (375.95, 36.13) -- (389.83, 36.13);
\end{scope}
\begin{scope}
\path[clip] (  0.00,  0.00) rectangle (505.89, 72.27);
\definecolor{drawColor}{gray}{0.40}

\path[draw=drawColor,line width= 0.4pt,line join=round,line cap=round] (380.12, 36.13) -- (385.67, 36.13);

\path[draw=drawColor,line width= 0.4pt,line join=round,line cap=round] (382.89, 33.36) -- (382.89, 38.91);
\end{scope}
\begin{scope}
\path[clip] (  0.00,  0.00) rectangle (505.89, 72.27);
\definecolor{drawColor}{RGB}{31,120,180}

\path[draw=drawColor,line width= 0.9pt,dash pattern=on 1pt off 3pt on 4pt off 3pt ,line join=round] (426.60, 36.13) -- (440.48, 36.13);
\end{scope}
\begin{scope}
\path[clip] (  0.00,  0.00) rectangle (505.89, 72.27);
\definecolor{drawColor}{RGB}{31,120,180}

\path[draw=drawColor,line width= 0.4pt,line join=round,line cap=round] (431.58, 34.17) -- (435.50, 38.10);

\path[draw=drawColor,line width= 0.4pt,line join=round,line cap=round] (431.58, 38.10) -- (435.50, 34.17);
\end{scope}
\begin{scope}
\path[clip] (  0.00,  0.00) rectangle (505.89, 72.27);
\definecolor{drawColor}{RGB}{0,0,0}

\node[text=drawColor,anchor=base west,inner sep=0pt, outer sep=0pt, scale=  1.00] at ( 46.92, 32.69) {AE};
\end{scope}
\begin{scope}
\path[clip] (  0.00,  0.00) rectangle (505.89, 72.27);
\definecolor{drawColor}{RGB}{0,0,0}

\node[text=drawColor,anchor=base west,inner sep=0pt, outer sep=0pt, scale=  1.00] at ( 88.66, 32.69) {Buzhash};
\end{scope}
\begin{scope}
\path[clip] (  0.00,  0.00) rectangle (505.89, 72.27);
\definecolor{drawColor}{RGB}{0,0,0}

\node[text=drawColor,anchor=base west,inner sep=0pt, outer sep=0pt, scale=  1.00] at (156.06, 32.69) {FSC};
\end{scope}
\begin{scope}
\path[clip] (  0.00,  0.00) rectangle (505.89, 72.27);
\definecolor{drawColor}{RGB}{0,0,0}

\node[text=drawColor,anchor=base west,inner sep=0pt, outer sep=0pt, scale=  1.00] at (204.48, 32.69) {Gear};
\end{scope}
\begin{scope}
\path[clip] (  0.00,  0.00) rectangle (505.89, 72.27);
\definecolor{drawColor}{RGB}{0,0,0}

\node[text=drawColor,anchor=base west,inner sep=0pt, outer sep=0pt, scale=  1.00] at (255.14, 32.69) {MII};
\end{scope}
\begin{scope}
\path[clip] (  0.00,  0.00) rectangle (505.89, 72.27);
\definecolor{drawColor}{RGB}{0,0,0}

\node[text=drawColor,anchor=base west,inner sep=0pt, outer sep=0pt, scale=  1.00] at (297.42, 32.69) {PCI};
\end{scope}
\begin{scope}
\path[clip] (  0.00,  0.00) rectangle (505.89, 72.27);
\definecolor{drawColor}{RGB}{0,0,0}

\node[text=drawColor,anchor=base west,inner sep=0pt, outer sep=0pt, scale=  1.00] at (342.50, 32.69) {Rabin};
\end{scope}
\begin{scope}
\path[clip] (  0.00,  0.00) rectangle (505.89, 72.27);
\definecolor{drawColor}{RGB}{0,0,0}

\node[text=drawColor,anchor=base west,inner sep=0pt, outer sep=0pt, scale=  1.00] at (397.06, 32.69) {RAM};
\end{scope}
\begin{scope}
\path[clip] (  0.00,  0.00) rectangle (505.89, 72.27);
\definecolor{drawColor}{RGB}{0,0,0}

\node[text=drawColor,anchor=base west,inner sep=0pt, outer sep=0pt, scale=  1.00] at (447.71, 32.69) {SeqCDC};
\end{scope}
\end{tikzpicture}
\vspace{-2em}
  \caption{Deduplication ratios at configured targets; note that the y-axis scales differ between plots.}
  \label{fig:dedup_overview}
\end{figure*}

\begin{figure*}[t]
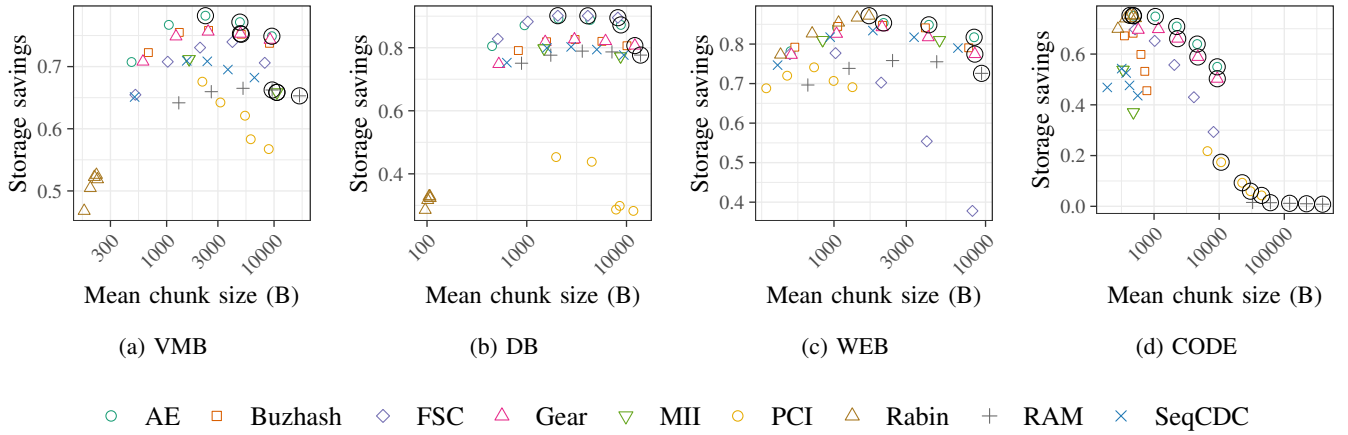

\centering
\subfloat[VMB]{\scalebox{0.85}{\input{fig/dedup_pareto_vmb}}%
    \label{fig:dedup_pareto:vmb}}
\hfil
\subfloat[DB]{\scalebox{0.85}{\input{fig/dedup_pareto_db}}%
    \label{fig:dedup_pareto:db}}
\hfil
\subfloat[WEB]{\scalebox{0.85}{\input{fig/dedup_pareto_web}}%
    \label{fig:dedup_pareto:web}}
\hfil
\subfloat[CODE]{\scalebox{0.85}{\input{fig/dedup_pareto_code}}%
    \label{fig:dedup_pareto:code}}

\vspace{-1.6em}
\begin{tikzpicture}[x=1pt,y=1pt]
\definecolor{fillColor}{RGB}{255,255,255}
\path[use as bounding box,fill=fillColor,fill opacity=0.00] (0,0) rectangle (505.89, 72.27);
\begin{scope}
\path[clip] (  0.00,  0.00) rectangle (505.89, 72.27);
\definecolor{fillColor}{RGB}{255,255,255}

\path[fill=fillColor] ( 25.98, 23.41) rectangle (479.91, 48.86);
\end{scope}
\begin{scope}
\path[clip] (  0.00,  0.00) rectangle (505.89, 72.27);
\definecolor{fillColor}{RGB}{255,255,255}

\path[fill=fillColor] ( 36.98, 28.91) rectangle ( 51.43, 43.36);
\end{scope}
\begin{scope}
\path[clip] (  0.00,  0.00) rectangle (505.89, 72.27);
\definecolor{drawColor}{RGB}{27,158,119}
\definecolor{fillColor}{RGB}{255,255,255}

\path[draw=drawColor,line width= 0.4pt,line join=round,line cap=round,fill=fillColor] ( 44.20, 36.13) circle (  1.96);
\end{scope}
\begin{scope}
\path[clip] (  0.00,  0.00) rectangle (505.89, 72.27);
\definecolor{fillColor}{RGB}{255,255,255}

\path[fill=fillColor] ( 76.73, 28.91) rectangle ( 91.19, 43.36);
\end{scope}
\begin{scope}
\path[clip] (  0.00,  0.00) rectangle (505.89, 72.27);
\definecolor{drawColor}{RGB}{217,95,2}
\definecolor{fillColor}{RGB}{255,255,255}

\path[draw=drawColor,line width= 0.4pt,line join=round,line cap=round,fill=fillColor] ( 82.22, 34.40) rectangle ( 85.70, 37.87);
\end{scope}
\begin{scope}
\path[clip] (  0.00,  0.00) rectangle (505.89, 72.27);
\definecolor{fillColor}{RGB}{255,255,255}

\path[fill=fillColor] (139.32, 28.91) rectangle (153.77, 43.36);
\end{scope}
\begin{scope}
\path[clip] (  0.00,  0.00) rectangle (505.89, 72.27);
\definecolor{drawColor}{RGB}{117,112,179}
\definecolor{fillColor}{RGB}{255,255,255}

\path[draw=drawColor,line width= 0.4pt,line join=round,line cap=round,fill=fillColor] (146.54, 33.68) --
	(149.00, 36.13) --
	(146.54, 38.59) --
	(144.08, 36.13) --
	cycle;
\end{scope}
\begin{scope}
\path[clip] (  0.00,  0.00) rectangle (505.89, 72.27);
\definecolor{fillColor}{RGB}{255,255,255}

\path[fill=fillColor] (184.07, 28.91) rectangle (198.53, 43.36);
\end{scope}
\begin{scope}
\path[clip] (  0.00,  0.00) rectangle (505.89, 72.27);
\definecolor{drawColor}{RGB}{231,41,138}
\definecolor{fillColor}{RGB}{255,255,255}

\path[draw=drawColor,line width= 0.4pt,line join=round,line cap=round,fill=fillColor] (191.30, 39.19) --
	(193.94, 34.61) --
	(188.66, 34.61) --
	cycle;
\end{scope}
\begin{scope}
\path[clip] (  0.00,  0.00) rectangle (505.89, 72.27);
\definecolor{fillColor}{RGB}{255,255,255}

\path[fill=fillColor] (230.73, 28.91) rectangle (245.18, 43.36);
\end{scope}
\begin{scope}
\path[clip] (  0.00,  0.00) rectangle (505.89, 72.27);
\definecolor{drawColor}{RGB}{102,166,30}
\definecolor{fillColor}{RGB}{255,255,255}

\path[draw=drawColor,line width= 0.4pt,line join=round,line cap=round,fill=fillColor] (237.96, 33.08) --
	(240.60, 37.66) --
	(235.31, 37.66) --
	cycle;
\end{scope}
\begin{scope}
\path[clip] (  0.00,  0.00) rectangle (505.89, 72.27);
\definecolor{fillColor}{RGB}{255,255,255}

\path[fill=fillColor] (272.85, 28.91) rectangle (287.30, 43.36);
\end{scope}
\begin{scope}
\path[clip] (  0.00,  0.00) rectangle (505.89, 72.27);
\definecolor{drawColor}{RGB}{230,171,2}

\path[draw=drawColor,line width= 0.4pt,line join=round,line cap=round] (280.07, 36.13) circle (  1.96);
\end{scope}
\begin{scope}
\path[clip] (  0.00,  0.00) rectangle (505.89, 72.27);
\definecolor{fillColor}{RGB}{255,255,255}

\path[fill=fillColor] (315.93, 28.91) rectangle (330.39, 43.36);
\end{scope}
\begin{scope}
\path[clip] (  0.00,  0.00) rectangle (505.89, 72.27);
\definecolor{drawColor}{RGB}{166,118,29}

\path[draw=drawColor,line width= 0.4pt,line join=round,line cap=round] (323.16, 39.19) --
	(325.80, 34.61) --
	(320.52, 34.61) --
	cycle;
\end{scope}
\begin{scope}
\path[clip] (  0.00,  0.00) rectangle (505.89, 72.27);
\definecolor{fillColor}{RGB}{255,255,255}

\path[fill=fillColor] (367.63, 28.91) rectangle (382.09, 43.36);
\end{scope}
\begin{scope}
\path[clip] (  0.00,  0.00) rectangle (505.89, 72.27);
\definecolor{drawColor}{gray}{0.40}

\path[draw=drawColor,line width= 0.4pt,line join=round,line cap=round] (372.08, 36.13) -- (377.63, 36.13);

\path[draw=drawColor,line width= 0.4pt,line join=round,line cap=round] (374.86, 33.36) -- (374.86, 38.91);
\end{scope}
\begin{scope}
\path[clip] (  0.00,  0.00) rectangle (505.89, 72.27);
\definecolor{fillColor}{RGB}{255,255,255}

\path[fill=fillColor] (417.11, 28.91) rectangle (431.56, 43.36);
\end{scope}
\begin{scope}
\path[clip] (  0.00,  0.00) rectangle (505.89, 72.27);
\definecolor{drawColor}{RGB}{31,120,180}

\path[draw=drawColor,line width= 0.4pt,line join=round,line cap=round] (422.37, 34.17) -- (426.30, 38.10);

\path[draw=drawColor,line width= 0.4pt,line join=round,line cap=round] (422.37, 38.10) -- (426.30, 34.17);
\end{scope}
\begin{scope}
\path[clip] (  0.00,  0.00) rectangle (505.89, 72.27);
\definecolor{drawColor}{RGB}{0,0,0}

\node[text=drawColor,anchor=base west,inner sep=0pt, outer sep=0pt, scale=  1.00] at ( 56.93, 32.69) {AE};
\end{scope}
\begin{scope}
\path[clip] (  0.00,  0.00) rectangle (505.89, 72.27);
\definecolor{drawColor}{RGB}{0,0,0}

\node[text=drawColor,anchor=base west,inner sep=0pt, outer sep=0pt, scale=  1.00] at ( 96.69, 32.69) {Buzhash};
\end{scope}
\begin{scope}
\path[clip] (  0.00,  0.00) rectangle (505.89, 72.27);
\definecolor{drawColor}{RGB}{0,0,0}

\node[text=drawColor,anchor=base west,inner sep=0pt, outer sep=0pt, scale=  1.00] at (159.27, 32.69) {FSC};
\end{scope}
\begin{scope}
\path[clip] (  0.00,  0.00) rectangle (505.89, 72.27);
\definecolor{drawColor}{RGB}{0,0,0}

\node[text=drawColor,anchor=base west,inner sep=0pt, outer sep=0pt, scale=  1.00] at (204.03, 32.69) {Gear};
\end{scope}
\begin{scope}
\path[clip] (  0.00,  0.00) rectangle (505.89, 72.27);
\definecolor{drawColor}{RGB}{0,0,0}

\node[text=drawColor,anchor=base west,inner sep=0pt, outer sep=0pt, scale=  1.00] at (250.68, 32.69) {MII};
\end{scope}
\begin{scope}
\path[clip] (  0.00,  0.00) rectangle (505.89, 72.27);
\definecolor{drawColor}{RGB}{0,0,0}

\node[text=drawColor,anchor=base west,inner sep=0pt, outer sep=0pt, scale=  1.00] at (292.80, 32.69) {PCI};
\end{scope}
\begin{scope}
\path[clip] (  0.00,  0.00) rectangle (505.89, 72.27);
\definecolor{drawColor}{RGB}{0,0,0}

\node[text=drawColor,anchor=base west,inner sep=0pt, outer sep=0pt, scale=  1.00] at (335.89, 32.69) {Rabin};
\end{scope}
\begin{scope}
\path[clip] (  0.00,  0.00) rectangle (505.89, 72.27);
\definecolor{drawColor}{RGB}{0,0,0}

\node[text=drawColor,anchor=base west,inner sep=0pt, outer sep=0pt, scale=  1.00] at (387.59, 32.69) {RAM};
\end{scope}
\begin{scope}
\path[clip] (  0.00,  0.00) rectangle (505.89, 72.27);
\definecolor{drawColor}{RGB}{0,0,0}

\node[text=drawColor,anchor=base west,inner sep=0pt, outer sep=0pt, scale=  1.00] at (437.06, 32.69) {SeqCDC};
\end{scope}
\end{tikzpicture}
\vspace{-2em}
\caption{Empirical mean chunk size and metadata-adjusted storage savings under \qty{64}{B} metadata cost per emitted chunk. Each plot uses its own axis
limits; x-axes are logarithmic; rings mark configurations nondominated among the evaluated points.}
\label{fig:dedup_pareto}
\end{figure*}

The degree of achievable deduplication is a key metric for evaluating CDC algorithms.
It determines how effectively an algorithm identifies and eliminates redundant data,
thereby reducing storage requirements.
We compare the selected algorithms by deduplication ratio and examine how dataset characteristics, such as compressibility and repeated byte patterns, relate to this metric.
The deduplication performance of CDC algorithms is typically measured using the following metric:
$$
\text{Dedup. ratio}=1-\frac{S_{\mathrm{dedup}}}{S_{\mathrm{original}}},
$$
where $S_{\mathrm{original}}$ represents the original size of the dataset and $S_{\mathrm{dedup}}$ represents the \emph{deduplicated} size, which is equivalent to the total size of unique chunks.
The dedup. ratio indicates the fraction of storage space saved by eliminating redundant chunks.
For instance, a deduplication ratio of $0.4$ on a dataset of \qty{1}{GB}
indicates that the same dataset can be represented using \qty{600}{MB} of unique emitted chunks.
Note that this metric excludes any metadata overhead.
Moreover, in \Cref{tab:csd_full}, we show that achieved mean chunk sizes can deviate strongly from the configured target. Small chunks have implications on the potential storage savings as they produce greater metadata overhead.
We therefore complement our study with \Cref{fig:dedup_pareto}, which plots the metadata-adjusted storage savings against the achieved mean chunk size and highlights the Pareto frontier of measured configurations.
For this figure, we employ the more practically oriented metric of \emph{storage savings}:
$$
\text{Storage savings}=1-\frac{S_{\mathrm{dedup}}+\beta N}{S_{\mathrm{original}}},
$$
where $N$ represents the total number of emitted chunks and $\beta$ denotes the metadata cost per chunk.
Fingerprint, indexing, and placement designs can produce different overheads. 
We model metadata overhead according to Debnath~\etal~\cite{chunkstash} and assume 64 bytes per emitted chunk (\ie, $\beta=64$): 20-byte SHA-1 fingerprints with an additional 44 bytes to contain chunk length and location.

\subsection{Results and Key Observations}
\label{sec:dedup:results}

Deduplication ratios generally decline as
the target chunk size grows (\Cref{fig:dedup_overview}). Rabin attains the highest ratios and the flattest
curves, but this mostly reflects its unresponsiveness to the configured target.
On DB its mean chunk size grows only from \qty{95}{B} to \qty{107}{B} as $\mu$ increases
from \qty{512}{B} to \qty{8}{KiB}, and on CODE it reaches just \qty{517}{B} at $\mu = 8192$
(\Cref{tab:csd_full}).
AE, RAM, and MII degrade faster because they reach larger targets by
scaling a horizon or window (e.g., $h = 7936$ for AE at \qty{8}{KiB}). This makes
cut-points depend on longer stretches of content, so we expect the trend to
continue beyond the tested range. MII's curves change in a single step, since
only $w = 6$ and $w = 7$ fall within our target range.
FSC performs well because every target is a multiple of the 512-byte TAR
block~\cite{gnu135Basic}, and VM images are similarly block-organized. Fixed
boundaries therefore often align with unchanged blocks across versions. RAM and
PCI fail on CODE, where zero padding in block-aligned archives suppresses
cut-points and inflates their mean chunk sizes to tens or hundreds of kilobytes.

Plotting metadata-adjusted storage savings against the empirical mean chunk size removes Rabin's apparent advantage (\Cref{fig:dedup_pareto}). 
At \qty{64}{B} per chunk, metadata alone costs 60--67\,\% of the original
size on DB, lowering Rabin's worst-case savings to around \qty{30}{\%}. 
RAM and PCI fail in the opposite direction: their oversized chunks miss most redundancy, with storage savings very close to 0 on CODE. 
With the exception of FSC on DB, AE achieves both the highest best-case savings and the
highest worst-case savings. 
Raw deduplication ratios thus favor algorithms that silently produce small chunks.
Once granularity and metadata are accounted for, AE offers the most favorable
and stable tradeoff, with Gear as the strongest hash-based alternative.

\section{Discussion}\label{sec:discussion}
We synthesize our results into guidance for selecting base chunking algorithms within the tested workloads and target range.
Our comparisons isolate the base chunking algorithm; they do not rank composite algorithms or end-to-end deduplication systems, whose behavior also depends on size thresholds, normalization, caching, fingerprinting, and storage organization.

\subsection{Summary of Results}

Our benchmark indicates that traditional hash-based algorithms, such as Rabin and Gear, still offer competitive deduplication efficacy, even when accounting for achieved mean chunk sizes.
However, we find that AE attains similar or better performance while closely adhering to the configured target chunk size and exhibiting low chunk-size variance.
Our experiments show that, depending on the structure of the dataset, Rabin can produce pathologically small chunks. These yield impressive deduplication ratios, but once metadata overhead is accounted for, a different picture emerges (see \Cref{fig:dedup_pareto:vmb,fig:dedup_pareto:db}).
We also discussed in \Cref{sec:dedup:results} that AE's deduplication may degrade when operating with target chunk sizes larger than those considered in our experiments, owing to the scaling-window problem. In those settings, Gear may represent a more robust alternative. Moreover, given the wide range of higher-level techniques that build upon hash-based algorithms, Gear remains an appealing choice for a variety of applications.
RAM, which has the highest measured throughput on RAND at \qty{2}{KiB}, proves sensitive to structured input:
On highly compressible data, such as the source-code archives in CODE,
the algorithm fails to find chunk boundaries, resulting in the generation of pathologically large chunks.
MII and PCI were proposed specifically for data synchronization applications, where chunk-size variance is of lesser importance.
Nevertheless,
in our storage-oriented evaluation, we do not observe clear advantages over alternatives such as AE or Gear,
which offer comparable or higher deduplication and throughput.
Likewise, the more recently published Seq\-CDC shows no benefit in our evaluated settings over, for instance, Gear or AE.

\Cref{tab:summary} summarizes results at $\mu=2\,\mathrm{KiB}$ with two measures of chunk-size calibration and variability:
``Max. Err.'' measures the maximum error $\frac{|\bar{c}-\mu|}{\mu}$ across datasets;
``Max. CV'' measures the maximum coefficient of variation $\frac{s}{\bar{c}}$ across datasets.

\begin{table}[tb]
    \centering
	\caption{Summary of key success metrics, based on experiments configured at $\mu=2\,\text{KiB}$. Superscripts identify every dataset attaining the extremum: C = CODE, W = WEB, V = VMB, D = DB. Bold: best; underline: second best.}
	\begingroup
\setlength{\tabcolsep}{4pt}
\begin{tabular}{lrrrrr}
\toprule
Algorithm & Throughput & \multicolumn{2}{c}{Chunk Sizes} & \multicolumn{2}{c}{Storage Savings} \\
\cmidrule(lr){3-4}\cmidrule(lr){5-6}
 & (MiB/s) $\uparrow$ & Max. Err. $\downarrow$ & Max. CV $\downarrow$ & Min. $\uparrow$ & Max. $\uparrow$ \\
\midrule
Rabin & 195 & 0.95\textsuperscript{D} & \underline{5.45}\textsuperscript{D} & 0.33\textsuperscript{D} & \underline{0.85}\textsuperscript{W} \\
Buzhash & 311 & 0.69\textsuperscript{C} & 30.59\textsuperscript{D} & 0.60\textsuperscript{C} & 0.85\textsuperscript{W} \\
Gear & \underline{652} & \underline{0.47}\textsuperscript{D} & 30.79\textsuperscript{D} & \underline{0.66}\textsuperscript{C} & 0.84\textsuperscript{W} \\
AE & 650 & \textbf{0.12}\textsuperscript{V} & \textbf{0.31}\textsuperscript{V} & \textbf{0.71}\textsuperscript{C} & \textbf{0.89}\textsuperscript{D} \\
RAM & \textbf{886} & 58.04\textsuperscript{C} & 57.81\textsuperscript{C} & 0.01\textsuperscript{C} & 0.79\textsuperscript{D} \\
PCI & 242 & 10.03\textsuperscript{C} & 195.44\textsuperscript{D} & 0.09\textsuperscript{C} & 0.74\textsuperscript{W} \\
MII & 209 & 0.84\textsuperscript{C} & 46.33\textsuperscript{D} & 0.54\textsuperscript{C} & 0.81\textsuperscript{W} \\
SeqCDC & 595 & 0.82\textsuperscript{C} & 34.11\textsuperscript{D} & 0.53\textsuperscript{C} & 0.83\textsuperscript{W} \\
\bottomrule
\end{tabular}
\endgroup

	\label{tab:summary}
\end{table}

\subsection{Contrasting Results with Literature}
\label{sec:discussion:contrast}

While our study confirms many of the claims made by the original algorithms' developers,
we also encounter several discrepancies.
These differences arise from factors such as dataset characteristics,
experimental setups, and implementation details.
We therefore contrast our results with those from the literature and reflect on their implications.
 
\paragraph*{AE}
As shown in \cref{sec:analysis}, the formula by which the parameterization in AE is derived
does not reflect the actual behavior of the algorithm fully,
but has been disseminated in this form throughout several studies~\cite{ae,ae2,ellappan2021dynamic,fastcdc2,jin2023accelerating,ellappan2023smart},
and moreover, the open-source chunking algorithm evaluation platforms Destor~\cite{fu2015design} and DedupBench~\cite{liu2023dedupbench}.
Following our report, DedupBench corrected its AE parameterization~\cite{dedupbenchIssue}.
According to our own experiments on RAND, the previously established formula renders chunks \qty{13}{\%} smaller than the target when it is set to \qty{512}{B}, 
and \qty{39}{\%} smaller with a target of \qty{8}{KiB}.
Because smaller chunks can improve deduplication ratios while increasing chunk and metadata counts, earlier AE comparisons must be interpreted with this size shift in mind.
The 13--39\,\% range quantifies the mean-size deviation in our RAND experiments, not a bound on deduplication gains across other workloads.
If a reported advantage could be explained by the size shift, the comparison remains unresolved without a re-evaluation controlling for empirical mean chunk size. We do not claim a ranking reversal: reproducing the affected studies would require their datasets, minimum/maximum constraints, fingerprinting, and other configuration details, and is outside our scope.
In the original paper~\cite{ae}, the authors suggest that AE is superior to hash-based algorithms in terms of deduplication.
While AE shows superior performance in some of our experiments and is competitive in most,
its performance sometimes degrades with higher target chunk sizes.

\paragraph*{RAM}
Our results show that RAM performs poorly on low-entropy datasets, most notably in the case of CODE (see \Cref{tab:csd_full}).
The authors of RAM notice this flaw~\cite{ram}; however, the full extent of this shortcoming was not evident from prior experiments.
In their study, the authors evaluate RAM on high-entropy datasets that may not represent the lower-entropy inputs often encountered in data deduplication workloads.
Indeed, deduplication is only evaluated in the form of \enquote{bytes saved per second},
which counterintuitively conflates throughput with deduplication.
Moreover, the paper suggests chunk-size variances similar to AE, while none of our experiments on real-world datasets support this statement.
Recently, Al~Jarah~\etal~\cite{udayashankar2024impact} evaluate the impact of low entropy with synthetic datasets.
While imposing a maximum chunk size constraint, the authors quantify an algorithm's robustness to low-entropy data by measuring the rate of artificial boundaries.
Interestingly, in their study, RAM exhibits better robustness than Rabin or FastCDC, which stands in contrast to our results.
This divergence may be explained by their reliance on synthetic datasets and the maximum chunk size constraint.

\section{Related Work}\label{sec:relatedWork}
Previous works have studied the state-of-the-art of data deduplication and chunking algorithms~\cite{xia2016comprehensive,goel2023detailed,fu2025distributed}.
These studies primarily focus on theoretical discussions based on the original works that introduced the algorithms.
Liu~\etal~\cite{liu2023dedupbench} experimentally evaluate FSC, AE, Gear, and Rabin.
Interestingly, their benchmark considers two AE configurations,
with chunk boundaries determined by either local maxima (AE-Max) or minima (AE-Min).
AE-Min yields slightly higher deduplication efficacy than AE-Max.
However, the authors only evaluate on a single dataset based on VM images.
We argue that performance differences (on any metric) between AE-Max and AE-Min are
most likely
based on the dataset's data distribution.
We further note that their evaluation relied on an incorrect formula for the AE window size,
resulting in an average chunk size nearly half of the expected value (see \Cref{subsubsec:setup:parameters:ae}).
Most other experimental evaluations and comparisons of chunking algorithms are part of works introducing new algorithms~\cite{ae,ram,pci,ellappan2021dynamic}.
These works do not provide a consistent set of datasets or comparable evaluation methodologies for judging the algorithms against one another.
For example, Manogar and Abirami~\cite{ellappan2021dynamic} present the superior throughput of their algorithm.
In their measurements, however, they do not isolate the effects of the chunking algorithm
from the computationally expensive SHA-1 fingerprinting of the produced chunks.
Consequently, the resulting throughput is heavily skewed in favor of algorithms that generate a lower number of chunks.
Notably, their own algorithm produces the fewest chunks across all evaluated algorithms and datasets.
The authors of PCI~\cite{pci} compare their algorithm against Rabin, LMC, AE, RAM, and MII.
The datasets used in the experiments are artificial, based on sequences of zeros with random byte insertions or deletions in defined intervals.
In practice, chunking algorithms are applied to real-world data, which is not reflected in the datasets used by the authors.
For instance, the issue of low-entropy strings is disregarded within this experiment.
Recently, Al Jarah~\etal~\cite{udayashankar2024impact} benchmark AE, FastCDC, Rabin, and RAM on synthetic datasets simulating low-entropy data.
Interestingly, using our methodology, we reach different conclusions, which we briefly discuss in \Cref{sec:discussion:contrast}.
%
Viji and Revathy~\cite{viji2021comparative} compare specifically the throughput of Rabin, LMC, AE, RAM, and PCI on random datasets.
In their experiment, RAM significantly outperforms the other algorithms,
followed by AE and PCI;
our experimental study confirms this ordering on random data.

\section{Conclusion}\label{sec:conclusion}
We showed that CDC algorithms are best understood and evaluated modularly: a base algorithm supplies the cut-point rule, while high-level techniques such as normalization, size thresholds, caching, and skipping change its operating envelope. By isolating this base layer, benchmarking eight algorithms on realistic backup-oriented datasets, and calibrating parameters on uniformly random data before freezing them for realistic workloads, we expose effects that prior composite evaluations obscure. We correct a widely propagated AE parameterization error, derive tuning rules for RAM and MII, and show that behavior on uniformly random data does not necessarily transfer to realistic workloads. Within the evaluated workloads and target range, AE is the most robust of the tested base chunking algorithms, Gear remains a strong hash-based choice, RAM is fast but fragile on low-entropy inputs, and Buzhash, MII, PCI, and SeqCDC offer limited advantages for storage-oriented deduplication.
We hope that our framework supports future advances in \ac{CDC} algorithms and enables more systematic evaluations of high-level techniques under the same modular discipline.

\section*{Acknowledgments}

The authors thank Asterios Katsifodimos for his valuable feedback and support. 
This work has been partly funded by the Dutch National NWO/TKI Science Grant BLOCK.2019.004 and the Federal Ministry of Research, Technology and Space of Germany within the project ``Open6GHub+'' (grant number: 16KIS2407).
Generative AI tools were used to improve the linguistic clarity and quality of the manuscript and to assist in writing R code for visualizing existing data.
No scientific content, software, results, or conclusions were generated using AI. 
All content was reviewed and verified by the authors.

%

\bibliographystyle{IEEEtran}
\bibliography{ref}

\clearpage
\appendices
\crefalias{section}{appendix}
\crefalias{subsection}{appendix}

\section{Stochastic Analysis of AE}\label{app:ae}
\begin{definition}
	In AE, $i+h$ marks a cut-point if it is the first position for which $B_i\geq\max\{B_j\}_{j=i+1}^{i+h}$.
\end{definition}

In order for any byte $B_i$ at position $i$ to become subject to the matching condition,
all previous bytes $B_{j<i}$ must have failed the matching condition $B_j\geq\max\{B_k\}_{k=j+1}^{j+h}$.
However, these are not independent events.
The event $B_{i-1}<\max\{B_{j}\}_{j=i}^{i+h-1}$ influences the expectation of $B_{i},\ldots,B_{i+h}$, and therefore affects the probability $B_i<\max\{B_j\}_{j=i+1}^{i+h}$.
So does every previous condition back to $B_{i-h-1}<\max\{B_j\}_{j=i-h}^i$, as it influences the expected value of $B_i$.
This leads to a \textbf{chain of recursive conditional probabilities} that we express as

\begin{equation*}
\begin{split}
    \mathrm{P}(i+h\text{ marks cut-point})= \mathrm{P}(B_i\geq \max\{B_j\}_{j=i+1}^{i+h}\mid \\ 
    B_{i-1}<\max\{B_j\}_{j=i}^{i+h-1}\land \ldots\land
    B_{i-h-1}<\max\{B_j\}_{j=i-h}^{i})\cdot \\
    \ldots\cdot\mathrm{P}(B_2<\max\{B_j\}_{j=3}^{h+2}\mid
     B_1<\max\{B_j\}_{j=2}^{h+1})\cdot\\\mathrm{P}(B_1<\max\{B_j\}_{j=2}^{h+1}).
\end{split}
\end{equation*}

For the purpose of readability, we introduce the abbreviation $M_h := \max\{B_i\}_{i=1}^h$ to denote the maximum value among $h$ bytes. 
As we are assuming uniformly distributed random data, and bytes are discrete variables in $[0,255]$,
the probability for any byte $X$ to take any value $0\leq x\leq 255$ is $P(X=x)=1/256$.
Furthermore, $P(X<x)=\frac{x}{256}$, $P(X\geq x)=\frac{256-x}{256}$, $P(X<Y)=\frac{1}{2}$, and $P(X\leq Y)=\frac{257}{512}$.

\begin{lemma}
	The cumulative distribution function $F_M(m)=\mathrm{P}(\max\{B_i\}_{i=1}^h\leq m)$ follows
	\begin{equation*}
		F_{M_h}(m)=\mathrm{P}(B_1\leq m\land\ldots B_h\leq m)=\left(\frac{m+1}{256}\right)^h
	\end{equation*}
	
	for $0\leq m\leq 255$
	and can be used to express the probability mass function of $M_h$:

	\begin{equation}
	\begin{split}
    	\mathrm{P}(M_h = m) = F_{M_h}(m)-F_{M_h}(m-1)\\=\left(\frac{m+1}{256}\right)^h-\left(\frac{m}{256}\right)^h
	\end{split}
	    \label{eq:ae_mh_m}
	\end{equation}

	for $0\leq m\leq 255$.
	Further, a lower bound $x$ changes this distribution to

	\begin{equation*}
    	\mathrm{P}(M_h=m\mid M_h\geq x) = \left(\frac{m-x+1}{256-x}\right)^h-\left(\frac{m-x}{256-x}\right)^h
	\label{eq:eight}
	\end{equation*}
	
	for $0 \leq m \leq 255$ and $0 \leq x \leq 255$.
\end{lemma}

\begin{figure}[h]
    \centering
    \input{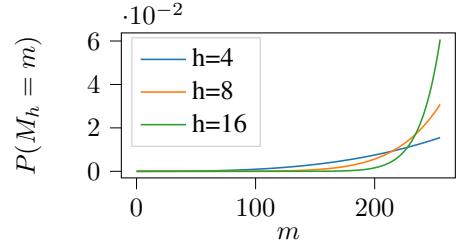}
    \caption{Probability mass function for $M_h$. Note, $m$ is discrete; line plots were chosen for readability. As $h$ increases, the expected value for $M_h$ converges to 255.}
    \label{fig:ae_pmf}
\end{figure}

Intuitively, as $h$ increases, the probability of observing small values of $M_h$ diminishes, since more bytes contribute to the maximum.
We illustrate this in \cref{fig:ae_pmf}. 
Further, $\lim_{h\to\infty} E(M_h) = 255$.
For instance, with $h=2000$, $\mathrm{P}(M_h=255)>0.999$.
That is, $\mathrm{P}(M_h=m\mid M_h\geq x)\geq\mathrm{P}(M_h=m)$ for any $0\leq x\leq 255$.

The practical implication of this is that when AE is used with target chunk sizes beyond \qty{2}{KB},
the approximation that $M_h\approx 255$ can be used to simplify the problem to

\begin{equation*}
    \mathrm{P}(i+h\text{ is cut-point}) = \mathrm{P}(B_i=255) = \frac{1}{256}.
\end{equation*}

On average, therefore, a chunk cut-point will be found after the evaluation of 256 bytes.
Hence, target chunk size for AE is $\mu\approx h+256$ for large values of $h$;
the parameterization is determined by $h\approx\mu-256$.
We continue with extended calculations for small target chunk sizes.

The probability for a random byte to fail the matching condition is

\begin{equation*}
    \mathrm{P}(X < \max\{B_i\}_{i=1}^{h}) = \sum_{m=0}^{255} \mathrm{P}(X < m)\cdot\mathrm{P}(\max\{B_i\}_{i=1}^{h}=m).
\end{equation*}

The probability in this form would apply to the first byte of the sequence, \ie, it is equivalent to $\mathrm{P}(B_1<\max\{B_i\}_{i=2}^{h+2})$.

The second byte $B_2$, which only gets evaluated under the pre-condition that $B_1<\max\{B_i\}_{i=2}^{h+1}$, faces the same condition as its precursor, but has a different probability as it is conditioned:

\begin{equation}
\begin{split}
    \mathrm{P}(B_2<\max\{B_i\}_{i=3}^{h+2}\mid B_1<\max\{B_i\}_{i=2}^{h+1})\\
    =\sum_{x=0}^{255} \mathrm{P}(B_2=x\mid B_1<\max\{B_i\}_{i=2}^{h+1})\cdot \\
    \mathrm{P}(x<\max\{B_i\}_{i=3}^{h+2}\mid B_1<\max\{B_i\}_{i=2}^{h+1})\\
    =\sum_{x=0}^{255} \Bigg( \mathrm{P}(B_2=x\mid B_1<\max\{B_i\}_{i=2}^{h+1}) \cdot \\
 \sum_{m=x+1}^{255} \mathrm{P}(\max\{B_i\}_{i=3}^{h+2}=m\mid B_1<\max\{B_i\}_{i=2}^{h+1}) \Bigg)
\end{split}
\label{eq:seven}
\end{equation}

The expectation of $B_2$ and the expectation of $\max\{B_i\}_{i=3}^{h+2}$ are shifted towards higher values, given that $B_1<\max\{B_i\}_{i=2}^{h+1}$.

The two conditional probabilities needed to solve \eqref{eq:seven} are less trivial. 
We are deriving $\mathrm{P}(B_2=x\mid B_1<\max\{B_i\}_{i=2}^{h+1})$ using Bayes' formula, which transforms the expression to

\begin{equation}
\frac{\mathrm{P}(B_1<\max\{B_i\}_{i=2}^{h+1}\mid B_2=x)\cdot\mathrm{P}(B_2=x)}{\mathrm{P}(B_1<\max\{B_i\}_{i=2}^{h+1})}.
\end{equation}

We derive the inversed conditional probability $\mathrm{P}(B_1<\max\{B_i\}_{i=2}^{h+1}\mid B_2=x)$ in the following lemma.

\begin{lemma}
The probability of a byte $B_1$ being smaller than the maximum value among the $h$ bytes in its horizon, given that one of the bytes $B_2$ is known, can be expressed as

	\begin{equation}
	\begin{split}
	    \mathrm{P}(x\geq \max\{B_i\}_{i=3}^{h+1})\cdot\mathrm{P}(B_1<x) \;\oplus\;
 	   \mathrm{P}(\max\{B_i\}_{i=3}^{h+1}>x) \\ 
 	   \cdot\mathrm{P}(B_1<\max\{B_i\}_{i=3}^{h+1})\mid \max\{B_i\}_{i=3}^{h+1}>x)
	\end{split}
	\end{equation}
	\label{lemma:ae_prob1}
\end{lemma}

\begin{proof}
	Given $B_2=x$, we know that $\max\{B_i\}_{i=2}^{h+1}\geq x$. That is, either $x$ is the maximum in this horizon and $x\geq \max\{B_i\}_{i=3}^{h+1}$, or the remaining bytes are $\max\{B_i\}_{i=3}^{h+1}>x$. 
\end{proof}

\begin{lemma}
	The probability distribution for the maximum value among $h$ bytes, given that the preceding byte is smaller than the maximum value in its horizon, 
	can be expressed as 
	
	$$
	\mathrm{P}(\max\{B_i\}_{i=3}^{h+2}=m\mid B_1<\max\{B_i\}_{i=2}^{h+1}),
	$$
	
	which can be evaluated as
	
	\begin{equation}
	\begin{split}
	\mathrm{P}(B_3=m \land B_j\leq m\;\forall j\in[4,h+1] \mid B_1<\max\{B_i\}_{i=2}^{h+1}) \cdot \\ 
	(h-1)\cdot
	    \mathrm{P}(B_{h+3}\leq m)+ \\
    	\mathrm{P}(B_{h+3}=m)\cdot\mathrm{P}(B_3\leq m\mid B_1<\max\{B_i\}_{i=2}^{h+1})^{h-1}
	\end{split}
	\end{equation}
	\label{lemma:ae_prob}
\end{lemma}

\begin{proof}
The expectation about a subset of bytes $\{B_i\}_{i=3}^{h+2} \cap \{B_i\}_{i=2}^{h+1} = \{B_i\}_{i=3}^{h+1}$ is influenced by the given condition,
\ie, all bytes in the target sequence except for its final byte $B_{h+2}$.
Intuitively, one byte will be $m$ while the remaining $h-1$ bytes of the intersecting horizon are $\leq m$.
By multiplying the chance for one of these bytes by $h-2$, we account for every byte in the intersecting horizon to be in place of the maximum value. 
In the formula, we used $B_3$ as an arbitrary placeholder, as any byte in the affected sequence yields identical chances.
Simultaneously, we also verify that the final (non-intersecting) byte $B_{h+3}\leq m$.
Finally, we account for the case that the final byte $B_{h+3}=m$ and all bytes in the conditioned sequence $\leq m$.
\end{proof}


Evaluating \eqref{eq:seven} using \cref{lemma:ae_prob1,lemma:ae_prob} 
accumulates six hierarchies of sum-loops. 
The complexity further depends on the size of $h$.
Even with low values for $h$, 
the evaluation turns out too computationally expensive (we interrupted the execution after two hours).
Therefore, we deem exact analytical solutions impractical.
We recommend an experimental approach for tuning $h$ for low target chunk sizes,
and making use of the approximation $h\approx\mu-256$ for target chunk sizes beyond \qty{2}{KB}.

\section{Stochastic Analysis of RAM}\label{app:ram}
\begin{definition}
Let $B_1,B_2,\ldots$ be independent bytes uniformly distributed on $\{0,\ldots,255\}$, and let $h\geq1$ be an integer.
RAM fixes $M_h=\max\{B_1,\ldots,B_h\}$ and cuts after the first subsequent byte at least as large as $M_h$.
Thus $T=\min\{t\geq1:B_{h+t}\geq M_h\}$ and the chunk length is $L=h+T$.
\end{definition}

Conditional on $M_h=m$, each subsequent byte succeeds with probability
$p_m=(256-m)/256$, independently of the preceding subsequent bytes.
Hence $T\mid M_h=m$ is geometric on $\{1,2,\ldots\}$, with
\begin{equation}
\mathrm{E}[T\mid M_h=m]=\frac{1}{p_m}=\frac{256}{256-m}.
\end{equation}
Averaging over the maximum distribution in \eqref{eq:ae_mh_m} gives a mixture of geometric waiting times:
\begin{equation}
\begin{split}
\mu(h)&=h+\mathrm{E}\bigl[\mathrm{E}[T\mid M_h]\bigr]\\
&=h+\sum_{m=0}^{255}\frac{256}{256-m}
\left[\left(\frac{m+1}{256}\right)^h
-\left(\frac{m}{256}\right)^h\right].
\end{split}
\end{equation}

The benchmark horizons were selected using the approximation in \Cref{eq:ram-approx}.
This approximation takes the reciprocal of the average success probability, rather than averaging the conditional waiting times.
By convexity, Jensen's inequality gives
$\mathrm{E}[1/p_{M_h}]\geq1/\mathrm{E}[p_{M_h}]$;
thus \eqref{eq:ram-approx} underestimates \eqref{eq:ram}.
For horizons 327 and 780 the differences are approximately \qty{32}{B} and \qty{6}{B}, respectively.

Solving \eqref{eq:ram-approx} and rounding gives the benchmark horizons $327,780,1792,3840,7936$ for targets $512,1024,2048,4096,8192$ bytes, respectively.
At the two smallest targets, both the exact-model predictions and the measured RAND means round to \qty{544}{B} and \qty{1030}{B}, respectively (\Cref{tab:csd_full}).
Minimizing $|\mu(h)-\mu|$ over positive integer horizons for each configured target $\mu$ instead selects 300 and 774 for the two smallest targets and the same benchmark horizons for the three larger targets; these are calculated parameter choices, not additional measurements.
Strictly increasing means permit binary search, with the smaller horizon selected on ties.

Finally, $\Pr[M_h=255]=1-(255/256)^h\to1$ implies $\mathrm{E}[T]\to256$ and $\mu(h)\approx h+256$.
At $h=1792$, this approximation differs from the exact mean by less than \qty{1}{B}.
Relating this IID model to measurements also requires accounting for implementation-specific boundary indexing and terminal chunks.

\end{document}